\documentclass[10pt]{article}
\usepackage[latin9]{inputenc}
\usepackage[letterpaper]{geometry}
\usepackage{float}
\usepackage{url}
\usepackage{amsmath}
\usepackage{amsthm}
\usepackage{amssymb}
\usepackage{graphicx}
\usepackage{setspace}
\usepackage[authoryear]{natbib}
\makeatletter
\theoremstyle{definition}
\newtheorem{defn}{\protect\definitionname}
\theoremstyle{plain}
\newtheorem{lem}{\protect\lemmaname}
\theoremstyle{plain}
\newtheorem{prop}{\protect\propositionname}
\theoremstyle{plain}
\newtheorem{cor}{\protect\corollaryname}

\usepackage{amsthm}\usepackage{enumerate}\usepackage{bm}\usepackage{calrsfs}\usepackage{amsfonts}
\usepackage{txfonts}

\usepackage{mathpazo}
\ifx\pdfoutput\@undefined\usepackage[usenames,dvips]{color}
\else\usepackage[usenames,dvipsnames]{color}
\definecolor{ChadDarkBlue}{rgb}{.1,0,.2}  
\definecolor{ChadBlue}{rgb}{.1,.1,.5}  
\definecolor{ChadRoyal}{rgb}{.2,.2,.8}  
\definecolor{ChadGreen}{rgb}{0,.4,0}    % Dark Green
\definecolor{ChadRed}{rgb}{.5,0,.5}  % purple

\let\LaTeXtitle\title
\renewcommand{\title}[1]{\LaTeXtitle{\color{ChadBlue}{\LARGE #1}}}

\let\LaTeX@startsection\@startsection 
\renewcommand{\@startsection}[6]{\LaTeX@startsection%
{#1}{#2}{#3}{#4}{#5}{\color{ChadBlue}\raggedright #6}} 

\renewcommand \thesection {\@arabic\c@section.}
\renewcommand\thesubsection   {\thesection\@arabic\c@subsection.}
\renewcommand\thesubsubsection{\thesubsection \@arabic\c@subsubsection.}

\usepackage{hyperref}
\hypersetup{
colorlinks=true,	% false: boxed links; true: colored links
citecolor=Maroon,	% color of links to bibliography
urlcolor=Maroon,	% color of external links
linkcolor=Maroon	% color of link
}

\makeatother

\providecommand{\corollaryname}{Corollary}
\providecommand{\definitionname}{Definition}
\providecommand{\lemmaname}{Lemma}
\providecommand{\propositionname}{Proposition}

\begin{document}

\title{Economic Implications of Blockchain Platforms }

\author{Jun Aoyagi\thanks{University of California at Berkeley, Department of Economics. E-mail:
\href{jun.aoyagi@berkeley.edu}{jun.aoyagi@berkeley.edu}, tel: (510)
541-3046. 610 Evans Hall, Berkeley, CA 94720.}\and Daisuke Adachi\thanks{Yale University, Department of Economics. E-mail: \href{daisuke.adachi@yale.edu}{daisuke.adachi@yale.edu}.
28 Hillhouse Avenue, New Haven, CT 06511.\protect \\
First draft: February 2018. This paper was initially circulated under
the title ``Fundamental Values of Cryptocurrencies and Blockchain
Technology.'' We appreciate the constructive comments from Kosuke
Aoki, Gadi Barlevy, Shin-ichi Fukuda, Christine Parlour, Giuseppe
Perrone, Susumu Sato, Yoko Shibuya, Harald Uhlig, Noriyuki Yanagawa,
Yingge Yan, and seminar participants at SWET, UC Berkeley, University
of Tokyo, and Yale University; and we are gratful to Claire Valgardson
for copy editing. We are very grateful to Hideaki Kajiura and EY Advisory
\& Consulting for sharing their data and experience from their wine
blockchain project. }}

\date{September 2018}
\maketitle
\begin{abstract}
In an economy with asymmetric information, the smart contract in the
blockchain protocol mitigates uncertainty. Since, as a new trading
platform, the blockchain triggers segmentation of market and differentiation
of agents in both the sell and buy sides of the market, it reconfigures
the asymmetric information and generates spreads in asset price and
quality between itself and traditional platform. We show that marginal
innovation and sophistication of the smart contract have non-monotonic
effects on the trading value in the blockchain platform, its fundamental
value, the price of cryptocurrency, and consumers' welfare. Moreover,
a blockchain manager who controls the level of the innovation of the
smart contract has an incentive to keep it lower than the first best
when the underlying information asymmetry is not severe, leading to
welfare loss for consumers.
\end{abstract}
\newpage{}

\section{\setstretch{1.5}Introduction}

Since Bitcoin was proposed by \citet{nakamoto2008bitcoin}, the notion
of the blockchain has gone viral as a new, innovative way to manage
information. It provides a decentralized, public information management
system in which data can be recorded as valid only if a consensus
has been reached. Moreover, Ethereum, the second-largest blockchain,
has invented a protocol to implement \textit{a smart contract}---one
that is executed automatically based on specified conditions without
any centralized authorizations (\citealp{szabo1997formalizing}).
We can exploit this protocol to exchange assets, products, and information.
For example, many blockchain-based trading platforms have been launched,
such as those for foods (EY Advisory \& Consulting, Walmart), jewelry
(HyperLedger), arts and photography (Kodac), security (tZero), and
cryptocurrency (Waves, IDEX, Steller, Oasis, OKEx, Cashaa, and more).

In spite of this growth, the academic research on these topics is
still in its infancy. We contribute to the literature by proposing
a simple yet intuitive theory that explores the economic implications
of blockchain technology. In accordance with previous research (see
the next subsection), our primary focus is on the blockchain as a
new platform for exchanging goods and assets. Given that the technology
aims to improve information management, we consider an asymmetric
information problem regarding the assets traded. Moreover, since the
blockchain works as a new platform and is operated in parallel with
a traditional exchange with no blockchains, it has the features of
a multi-platform economy with two-sided markets, as described in the
field of industrial organization (IO). We investigate how innovation
in blockchain technology affects the segmentation of the trading platforms,
the price and quality of the assets traded, information asymmetry,
and consumers' welfare. We also define the fundamental values of the
blockchain platform and its attached cryptocurrency.

The smart contract is one of the most innovative aspects of the blockchain
system, which differentiates it from the traditional exchange protocol
with cash or credit. In traditional exchange, there is no way to eliminate
the asymmetric information \textit{a priori}, and the possibilities
of adverse selection and market breakdown are omnipresent. To mitigate
this problem, a typical economy relies on intermediations by a third
party, such as banks, insurance providers, and central securities
depositories, to offload the risks. In contrast, a blockchain transaction
is immune from information asymmetry due to the security mechanism
hard-wired into the protocol. As discussed in Section \ref{sec:Technology-Overview:-Cryptocurre},
transaction information stored in the blockchain is protected from
tampering, that is, rewriting the transaction record comes at a prohibitively
large cost. Crucially, Ethereum allows the transactions to be executed
based on sophisticated scripts; users can write a code on the blockchain
that describes the specific conditions they wish them to fulfill.
Hence, the transaction can be \textit{state contingent}, and the validity
of ``state information'' is highly credible. This highlights the
difference of the blockchain protocol from credit as a record-keeping
method, since the latter is not responsible for the actual transfer
and quality of goods, while both of them are automatically guaranteed
on Ethereum.

We consider non-atomic sellers and buyers who decide what type of
transaction platform to use to exchange assets whose quality (high
or low) is unknown to buyers. We define the smart contract in our
economy by claiming that a transaction by means of the blockchain
technology bears less information asymmetry. The traditional market
(cash-market or $C$-market) cannot detect low-quality assets, and
buyers face severe quality uncertainty. On the blockchain platform
($B$-market), in contrast, the low-quality assets can be detected
and excluded before trading occurs with some probability $\theta$,
which we call the security level (Appendix \ref{subs:motive} provides
a couple of examples as a micro-foundation for $\theta$). At first,
we take $\theta$ as an exogenous parameter, whereas in Subsection
\ref{subsec:Fundamental-Value-of} and onward, we study a manager
of the blockchain platform who controls $\theta$. Our main focus
is on how $\theta$ affects the activity of the entire economy.

 As the literature on two-sided markets suggests, the first direct
consequence of the emergence of the blockchain platform (a positive
$\theta$) is the differentiation and segmentation of \textit{both}
the sell and buy sides of the market. This segmentation is accompanied
by endogenous spreads in the quality and price of assets between the
blockchain and traditional platforms. Then, we find that marginal
innovation in the security level $\theta$ has non-monotonic effects
on the transaction \textit{value} in the blockchain platform, the
fundamental value of this platform, the price of cryptocurrency, and
consumers' welfare. That is, a more secure blockchain platform does
not necessarily induce more active transactions and better allocation
for consumers. 

A higher $\theta$ directly improves the quality of assets supplied
in the $B$-market, since the smart contract precludes a certain fraction
of bad assets. In addition, quality in the $B$-market is endogenously
amplified by the general equilibrium effect. The higher quality induces
a higher price of the assets traded through the $B$-market due to
a higher expected return.\footnote{The ``price of assets in the $j$-market'' is valued in terms of
cash, not cryptocurrency. Accordingly, we can see cryptocurrency as
an asset, and its price is also valued in terms of cash.} On the supply side, sellers of low-quality assets confront the price-liquidity
(rejection) tradeoff. They can obtain a higher return from trading
the asset in the blockchain platform, but at the risk of being rejected
and ending up consuming their own low-quality asset. On the other
hand, if a seller has a high-quality asset, the net return from selling
it in the $B$-market monotonically increases as the bid price goes
up with no fear of rejection. Thus, the reaction to innovation differs
depending on the nature of the seller's asset, endogenously boosting
the flow of high-quality assets into the $B$-market.

On the demand side, whether the higher security attracts more buyers
to the $B$-market depends on the relative rise in the assets' quality
versus the increase in the price, \emph{i.e.}, buyers face the traditional
price-quality tradeoff. We show that the improvement in the quality
is driven solely by the sell side's behavior, while the price change
is caused by the increase in the demand and decrease in the supply,
which leads to a larger increase in the price than in the quality.
As a result, a higher $\theta$ pushes the price up and reduces the
transaction volume, making the $B$-market ``an exclusive market
for a high-quality but expensive asset.'' 

If the primitive asymmetric information is not severe, it is easier
for buyers to give up the higher quality in the $B$-market and migrate
to the $C$-market to save the price cost. This implies that a higher
$\theta$ generates a larger decline in the transaction volume than
the increase in the price, leading to less activity in the $B$-market
as measured by the trading value.

To quantify the fundamental value of the blockchain itself, we also
consider a blockchain manager who proposes an \textit{ex-ante} contract
that enables traders to use the platform for a fee. Access to the
blockchain platform generates a strictly positive welfare gain for
market participants, which makes traders willing to pay. This positive
fee can be seen as the fundamental price of the blockchain technology,
and we show that its behavior has the same implications as the trading
value in the $B$-market and is non-monotonic in $\theta$. Therefore,
our model implies that the sophistication of the blockchain, measured
by a higher $\theta$, can reduce the technology's value as a trading
platform.

Based on this non-monotonicity, we discuss the optimal security level
$\theta$ for the platform manager. As mentioned earlier, a higher
$\theta$ can reduce the number of consumers who participate in the
$B$-market, the trading value, and the aggregate welfare gain for
consumers. This implies that the blockchain manager can charge only
a small fee, since, \textit{ex-ante}, each consumer expects that the
gain from participating in the exclusive $B$-market is small. This
gives the manager an incentive to keep $\theta$ lower than the first-best
level for consumers ($\theta=1$). In other words, she cares only
about the transaction value in the $B$-market and does not reckon
with the effect of $\theta$ on the activity in the $C$-market, making
her underestimate the benefit of an increase in $\theta$.\footnote{This is consistent with the literature on strategic management (\citealp{teece1986profiting};
\citealp{brandenburger1996value}) arguing that a firm may not adopt
innovation even though it improves consumers' welfare.}

This decline in the trading value and maximum possible fee tend to
occur when the asymmetric information is not severe and migration
is easier for consumers. In contrast to the conventional perspective,
our results indicate that the government should intervene in the market
to promote blockchain transactions when an information problem is
\textit{not} severe, while it does not need to meddle when it is severe. 

After a review of the related literature, Section 2 provides an overview
of the technology of the blockchain and cryptocurrency. Section 3
introduces the theoretical environment, while Section 4 analyzes comparative
statics to understand the effect of higher security in blockchain
technology. In Section 5, we propose the empirical hypotheses, and
Section 6 concludes the discussion. 

\subsection{Literature Review\label{subsec:Related-Literature}\label{Lit}}

The research on blockchain technology and cryptocurrencies (or FinTech,
in general) is expanding (see \citeauthor{harvey2016cryptofinance}
{[}\citeyear{harvey2016cryptofinance}{]} for a comprehensive review).
First, viewing the blockchain protocol as a new trading platform is
widely accepted. \citet{bartoletti2017empirical} provide empirical
evidence for the usage of the blockchain and the smart contract as
a platform. \citeauthor{chiu2017economics} (\citeyear{chiu2017economics},
\citeyear{chiu2018blockchain}) analyze the optimal design of the
blockchain to guarantee ``Delivery vs. Payment'' by considering
an economy with an intertemporal risk of settlement. \citet*{cong2018tokenomics}
develop a model in which the demand and price dynamics of tokens (cryptocurrency)
are driven by the size of the blockchain as a platform and its trading
needs. 

The blockchain can affect consumers' welfare through many channels.
According to \citet{cong2017blockchain}, its reduction of asymmetric
information promotes the entrance of firms and improves consumer welfare,
although its efficient record keeping makes it easier for firms to
collude. \citet*{malinova2017market} compare the possible degrees
of transparency of the private blockchain and find that the most transparent
setting maximizes consumers' welfare at the risk of front-running.\footnote{Users of blockchains can make the network private and limit information
transactions within a firm or a group of firms. This category of platforms
is called a \textquotedblleft closed-type\textquotedblright{} or \textquotedblleft permissioned\textquotedblright{}
blockchain. The public blockchain, in contrast, is called \textquotedblleft open-type\textquotedblright{}
or \textquotedblleft permissionless.\textquotedblright{}} \citet{khapko2016smart} focus on the optimal duration of the transactions
under counterparty risk and search friction to show that the optimal
implementation of settlement can improve welfare.\footnote{The feasibility of the blockchain implementation is another hot topic.
For example, \citet{biais2018blockchain} consider ``the folk theorem''
of the blockchain as a coordination game, and \citet*{aune2017footprints}
propose the hash-based protocol to address the issue that stems from
miners' incentive to delay the publication of the block. }  

Our model agrees with these studies on the fundamental effect of the
blockchain: it reduces the transaction cost by mitigating informational
problems. However, our economy, in which buyers face \textit{ex-ante}
quality uncertainty, highlights how the blockchain endogenously reconfigures
information asymmetry via quality differentiation and platform segmentation---both
of which are not analyzed in the literature---and how it affects
the value of the blockchain platform and consumers' welfare. 

The second strand of the literature, which emerged from \citet{akerlof1970market},
examines adverse selection. Authors such as \citet{kim2012endogenous},
\citet{guerrieri2014dynamic}, and \citet{chang2017adverse} show
that market segmentation leads to quality differences across markets.\footnote{Another dimension of segmentation is time, \emph{i.e.}, participants
can decide when to trade, as analyzed by \citet*{fuchs2016adverse},
\citet*{asriyan2017information}, and \citet{Fuchs2017}.} We see the blockchain as a new platform for trade, which exists alongside
the traditional cash market, and analyze the effect of segmentation
in the context of FinTech. Unlike the literature, in which the markets
are homogeneous \emph{per se}, our analyses propose that the different
structure of one market (\emph{e.g.}, the degree of security) affects
the entire economy. Also, these works do not investigate the manager
who optimally chooses the structure of her platform.

Our paper is also related to the literature in IO that analyzes endogenous
market structures and platform competition with two-sided markets.
The market segmentation and differentiation of agents have been analyzed
by \citet*{foucault2004competition}, \citet{rochet2006two}, \citet{damiano2008competing},
\citet{ambrus2009asymmetric}, and \citet{gabszewicz2014vertical},
although they do not study asymmetric information in the form of asset
quality uncertainty. \citet{yanelle1997banking} and \citet{halaburda2013platform}
consider competing platforms under asymmetric information, though
they focus on the information asymmetry between the platform and agents.
In line with these works, our platform manager controls the asymmetric
information between agents, while she does not incorporate some types
of externality due to the general equilibrium effects.\footnote{More broadly, a different security level of the blockchain $\theta$
in a market with asymmetric information can be interpreted as government
intervention, such as the purchase of assets on OTC after the recent
financial crisis. Most models of government intervention do not consider
the interaction of segmented markets. See, for example, \citet*{philippon2012optimal},
\citet{tirole2012overcoming}, and \citet*{chiu2016trading}.} 

Finally, the idea of the blockchain as a record-keeping method that
competes with the traditional cash is reminiscent of the concept of
``money as memory.'' How money and credit can substitute for each
other or coexist has been explored by \citet{kocherlakota1998money},
\citet{kocherlakota1998incomplete}, \citet{lagos2005unified}, \citet{rocheteau2005money},
\citet{camera2008another}, and \citet*{gu2014money}. However, credit
as an alternative payment method does not affect the quality of the
assets traded because it remains a record of a debtor, but not of
the assets' quality. Thus, the market segmentation in the literature
is intertemporal (day and night): a decentralized market comes first
to make some credit, and a centralized one comes afterwards to settle
repayment. As some evidence in Section \ref{sec:Empirical-Implications}
suggests, however, the blockchain platform triggers market segmentation
even in a static environment by keeping a record of assets' quality,
making transactions contingent on it, and differentiating both buyers
and sellers.

\section{Technology Overview: Blockchain Protocol\label{sec:Technology-Overview:-Cryptocurre}}

The blockchain can be seen as a novel way of managing and tracking
transaction information. Participants in a transaction (say, sellers
and buyers) possess private information about their state---how much
bitcoin they own or have already spent, the quality of the products
they sell, and so on. In the traditional world, we typically maintain
a ledger that records participants' state information in a centralized
manner, \emph{e.g.}, there is a bank as an intermediary. Bilateral
transactions with no intermediations by a credible third party incur
the risk of adverse selection due to asymmetric information or settlement
risk. 

In contrast, on the blockchain platform, the ledger is not held by
a particular entity, but is distributed across all participants in
the network. The \textit{distributed ledger system} requires the information
about the state of the economy to be a consensus among all the participants.
This highlights its first difference from traditional transactions,
in which only a centralized intermediary keeps track of the state
information.

Moreover, Ethereum allows complex scripts to be written to describe
the conditions under which the information is verified and recorded,
which implies that a transaction takes place only if the conditions
in the code are fulfilled \emph{(i.e.}, it is state contingent). This
is the crucial aspect that differentiates the blockchain from the
credit system (or credit cards) as a record-keeping method.\footnote{A warranty is an example of a similar system to the blockchain. Even
though warranties guarantee the quality of products, they have two
main differences. First, warranties are provided and executed only
if the product is transferred and the \textit{ex-post} quality is
verified, while the smart contract transfers the product only if the
quality is guaranteed. Thus, they differ in which aspect of the transaction,
the quality or the transfer of products, is contingent on the other.
Second, the execution of warranties comes at a significant cost for
consumers, while the smart contract never requires consumers to take
the cost because making a transaction serves as a quality certification.
See, for example \citet{lehmann1974consumer} and \citet{palfrey1983warranties}.}

In general, it is extremely difficult for one member of the network
to overturn the consensus. In the case of Bitcoin, for example, system
managers called miners leverage their computing power to solve a time-consuming
cryptographic problem. This process is called ``proof of work,''
and the miner who performs it fastest is entitled to add a new block
to the chain. Therefore, if a malicious agent attempts to add fraudulent
information to the transaction history, she must outpace all miners
in the network, which requires prohibitively high computing power.\footnote{There are several ways to reach a consensus, and different blockchains
adopt different processes. \citet{chiu2017economics} provide a theoretical
comparison of the efficiency of these methods. }

Once a set of transaction information forms a block, it is encrypted
by a hash function and passed to the next block to create a chain
of blocks. The output of the hash function becomes different if one
entity of input is different. Thus, revising a piece of information
in a chain requires the revision of all of the subsequent information
in the blockchain.\footnote{For example, if a state of the transaction up to $t$ is denoted by
$s^{t}=(s_{t-1},\cdots,s_{0})$, a block at $t$ records the information
of transactions at $t$ and the encrypted historical states, $S_{t}=(s_{t},h(s^{t}))$,
where $h$ is a hash function. Now, the next block at $t+1$ records
$S_{t+1}=(s_{t+1},h(S_{t}))$, and so on. If an agent wants to rewrite
the past state $s_{k}$ to $s'_{k}$, then she must change all the
blocks $S_{k},S_{k+1},\cdots,S_{t}$ because this attempt induces
the change in $S_{k}$, which triggers a change in $S_{k+1}$ because
$h(S_{k})$ with $s_{k}$ is not identical to that with $s'_{k}$,
and so on. } As a result, any attempts to benefit from modifying the existing
information is virtually impossible. That is to say, only relevant
information can be added to the blockchain, and it is free from tampering.

\section{Theoretical Framework}

Consider an economy with segmented markets that operate at $t=0$
and $1$. There is a continuum of risk-neutral buyers (consumers)
and sellers (producers), both characterized by the private value $\alpha\in[0,1]$.
$\alpha$ has a cumulative measure $F$, which will be assumed to
be uniform.\footnote{Imposing the same $F$ on both sellers' and buyers' $\alpha$ is for
tractability and is not essential in our analyses.} At date $t=0,$ each buyer is endowed with a certain amount of cash
$w$, draws $\alpha$, and partakes in markets to buy an asset. 

On the sell side, each seller is endowed with a unit asset with a
stochastic quality $q$, which is either high, $q=H$, or low, $q=L$,
with $\Pr(q=H)=\pi\in(0,1)$ and independent of $\alpha$. By the
law of large numbers, the economy-wide fraction of high-quality assets
is $\pi$, and that of low-quality assets is $1-\pi$. Also, risk-free
savings with a zero interest return are available.

For both sellers and buyers with a private value $\alpha$, the asset
yields the following (per capita) utility at date $t=1$:
\[
y(\alpha)=\begin{cases}
\alpha & \text{if \ensuremath{q=H}}\\
\phi\alpha & \text{if \ensuremath{q=L.}}
\end{cases}
\]
$\phi\in(0,1)$ is the primitive quality difference. If the asset
is low-quality, agents obtain only $\phi$ fraction of the utility.
Following the literature on market microstructure (such as \citealp*{glosten1985bid}),
agents can trade and hold, at most, one unit of asset and cannot short-sale. 

\subsection{Market Structure}

There are two trading platforms with different mediums of exchange.
One is the blockchain with cryptocurrency and the other is the traditional
market with fiat money. We give transactions via the blockchain platform
an index $j=B$ (blockchain) and those with fiat money an index $j=C$
(cash). 

\subsubsection*{Smart Contract}

To distinguish transactions via the blockchain from those through
the cash market, we define the \textit{smart contract} in our model
as follows.
\begin{defn}
$\theta$ fraction of low-quality assets that sellers intend to sell
in the $B$-market are detected and rejected by the blockchain mechanism.
\end{defn}
The parameter $\theta$ is called the ``security level'' of the
blockchain platform. We provide the motivating example and micro-foundation
of $\theta$ in Appendix \ref{sec:Appendix:-Motivating-Examples}.
The interpretation of $\theta$ can differ depending on the context
in which we apply this model. For instance, if the traded asset is
the cryptocurrency itself, such as Bitcoin, $\theta$ is the probability
that the mechanism will preclude attempted ``double spending.''
If the traded asset is consumption goods, as in the wine blockchain,
$\theta$ captures how intensively the contracts are made contingent
throughout the intermediations between producers and consumers.\footnote{Another interpretation of $\theta$ is in light of consensus quality
in \citet{cong2017blockchain}. Although the distributed ledger makes
a consensus on the state close to perfect by aggregating the reports
by system keepers, there is still noise and bias. We take a large
$\theta$ as a precise consensus. Also, \citet*{kroll2013economics}
consider a risk that the distributed ledger system goes wrong by group's
attempt to make a consensus. It is also shown by \citet{biais2018blockchain}
that a folk of a chain generates two (or more) different consensuses.
We capture these events by $\theta<1$.} At first, we take this value as given for participants in the market.
Later, in Subsection \ref{subsec:Fundamental-Value-of} and on, we
study a case wherein a manager of the blockchain can control this
security level.\footnote{We can incorporate insurance or third-party institutions that reduce
the risk of lemons in the $C$-market, which typically exist in the
real economy. One possible way to describe them is by introducing
a rejection probability $\theta_{C}$ in the $C$-market as well.
A more parsimonious way, which we follow, is to think of $1-\pi$
as the fraction of low-quality assets that stay in the economy even
after we ask third-party institutions to reject them, \emph{i.e.},
the existing insurance cannot cover all the assets.}

Moreover, to motivate agents to hold cryptocurrency, we introduce
the following restriction:
\begin{description}
\item [{Assumption$\ $1.\label{Assumption2.-If-a}}] To buy $k_{B}$ amount
of assets at price $P_{B}$ (in terms of cash) in the $B$-market,
a buyer must hold $P_{B}k_{B}/Q$ of cryptocurrency, where $Q$ is
the price of cryptocurrency in terms of cash (see the budget constraint
{[}\ref{eq:budget}{]}).
\end{description}
This assumption comes from the fact that the endowment is given by
cash. We call it the ``cryptocurrency in advance'' (CIA) constraint
in our model, and \citet*{schilling2018some} consider a similar formulation.\footnote{As explained in the introduction, this assumption captures the class
of cryptocurrencies that is used as a means of exchange in the blockchain
trading platform. The analyses in Subsection \ref{sec:Welfare-Analyses:-Fundamental}
provide a measure of the fundamental value of the blockchain, instead
of the value of the cryptocurrency, and hence does not need this assumption
regarding CIA. } It will be clear that the demand and pricing for cryptocurrency are
determined mostly outside of the asset trading market. Therefore,
by removing Assumption 1 and imposing it on sellers' behavior, we
can still analyze the other class of cryptocurrency platforms, such
as a part of Ethereum, in which the sellers must have cryptocurrency
to verify their authenticity. Also, we show that the fundamental price
of blockchain technology can be characterized even without Assumption
1.

\subsection{Optimal Behavior of Buyers}

A buyer with type $\alpha$ maximizes her expected consumption at
$t=1$, $V(\alpha)=E[c|\alpha]$, under the following budget constraints:
\begin{align}
w & \ge P_{C}k_{C}+Qb+s,\text{\ }\frac{Q}{P_{B}}b\ge k_{B},\label{eq:budget}\\
c & =y_{C}(\alpha)k_{C}+y_{B}(\alpha)k_{B}+s.\nonumber 
\end{align}
$k_{j}$ and $P_{j}$ represent the demand and price of the asset
at market $j$, $Q$ is the price of cryptocurrency, and $b$ is the
demand (quantity) for cryptocurrency. All prices are valued in terms
of cash. Thus, the price of assets traded in the $B$-market in terms
of cryptocurrency is $P_{B}/Q$. A risk-free saving option is denoted
by $s$. The definition of $y_{j}(\alpha)$ is given by (\ref{eq:yj})
below.

The constraints in the first line imply that the buyer allocates her
cash endowment to the purchase of the asset in the $C$-market and
cryptocurrency, and the latter is used to buy the asset in the $B$-market.
The purchase amount in the $B$-market is limited by her holdings
of cryptocurrency, as Assumption 1 suggests. As well, the agent can
stay inactive to get zero utility from her assets. The second line
shows the consumption level, in which, for $j\in\{B,C\}$,
\begin{align}
y_{j}(\alpha) & \equiv\tilde{\pi}_{j}\alpha\equiv[\pi_{j}+(1-\pi_{j})\phi]\alpha,\label{eq:yj}\\
\pi_{j} & \equiv\Pr(q=H\text{ in Market-\ensuremath{j}}).
\end{align}
Hence, $y_{j}$ represents the expected private return adjusted by
the risk of lemons in market $j$, which is denoted by $1-\pi_{j}$. 

Because of the risk neutrality and linearity of $y$, splitting order
into two markets is not optimal.\footnote{Only the buyers on the threshold (defined below) can split the order,
but we simplify our discussion by assuming a tie-breaking rule that
indifferent agents trade in the $B$-market.} Thus, the demand always hits its upper limit ($k_{j}=1$), and the
CIA constraint is binding. The expected return from purchasing assets
in each market ($V_{B}$, $V_{C}$) and that from staying inactive
($V_{0}$) are given by
\begin{align*}
V_{j}(\alpha)=\begin{cases}
\tilde{\pi}_{j}\alpha-P_{j} & \text{if \ensuremath{j\in\{B,C\}}}\\
0 & \text{if \ensuremath{j=0.}}
\end{cases}
\end{align*}
We subtract $w$ from the equations above because it does not affect
the equilibrium behavior.

To solve the venue-choice problem, we guess the following\footnote{See \citet{gabszewicz2014vertical} for a similar structure, in which
they state these as assumptions, while we derive them endogenously.}:
\begin{equation}
\frac{P_{B}}{\tilde{\pi}_{B}}>\frac{P_{C}}{\tilde{\pi}_{C}},\text{\ensuremath{\pi_{B}>\pi_{C},}}\label{guess}
\end{equation}
which will be shown to be a unique equilibrium. Intuitively, $P_{j}/\tilde{\pi}_{j}$
is a normalized price and represents the cutoff of $\alpha$ that
generates indifference between buying in the $j$-market and staying
inactive.\footnote{Buyers' behavior can be seen as the model of vertical differentiation,
such as the one provided in Chapter 2 of \citet{tirole1988theory}.} It indicates a positive measure of traders with relatively high (resp.
low) $\alpha$ who wish to go to the $B$-market (resp. $C$-market).\footnote{If (\ref{guess}) does not hold, the $B$-market becomes too attractive
to guarantee the coexistence.} Indeed, under (\ref{guess}), the optimal behavior of buyers with
type $\alpha$ is determined by the cutoff $\alpha^{*}$ such that
\[
\alpha^{*}\equiv\frac{P_{B}-P_{C}}{\tilde{\pi}_{B}-\tilde{\pi}_{C}}.
\]
Figure \ref{Fig_buyer} plots returns, $V_{i}$, against $\alpha$
and shows the cutoffs for the optimal behavior. Namely, it is optimal
for type $\alpha$ buyers to (i) buy one unit of the asset in the
$B$-market if $\alpha\ge\alpha^{*}$, (ii) in the $C$-market if
$\alpha\in[\frac{P_{C}}{\tilde{\pi}_{C}},\alpha^{*})$, and (iii)
stay inactive otherwise. 

Intuitively, each buyer faces a price-quality tradeoff, \emph{i.e.},
the $B$-market provides higher quality and expected returns, but
charges a higher price. Note that the gain from a higher $\pi_{j}$
is multiplied by $\alpha$, while the cost is constantly $P_{j}$.
Hence, the $B$-market looks more attractive for high-$\alpha$ buyers.
\begin{figure}[H]
\begin{center}\caption{Returns for Buyers}\label{Fig_buyer}

\includegraphics[scale=0.35]{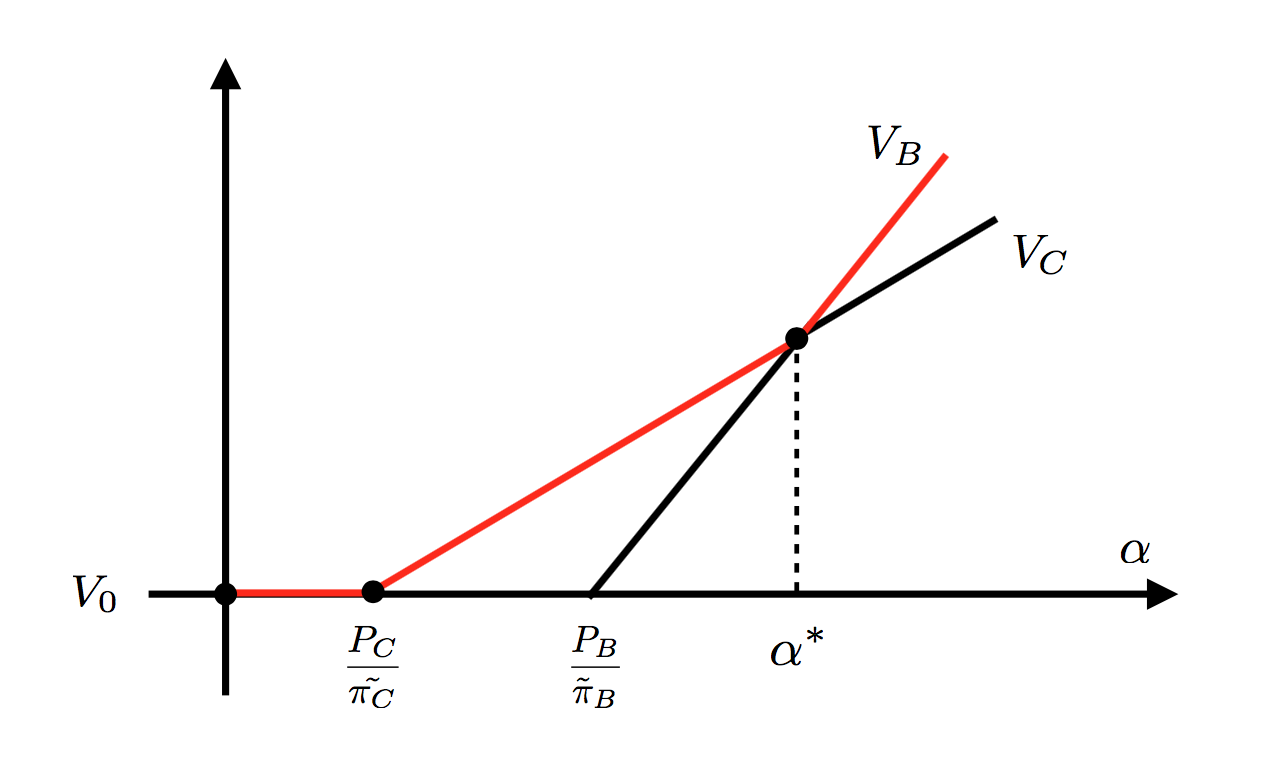}

\end{center}
\end{figure}

By aggregating along $\alpha$, the total demand in each market is
\begin{align}
K_{j}^{D} & =\begin{cases}
1-F(\alpha^{*}) & \text{for \ensuremath{j=B,}}\\
F(\alpha^{*})-F\left(\frac{P_{C}}{\tilde{\pi}_{C}}\right) & \text{for \ensuremath{j=C}.}
\end{cases}\label{KBD}
\end{align}
 The uniform $F$ allows us to derive inverse demand functions:\footnote{As for the form of $F$, the uniform assumption is restrictive in
this model. The equilibrium is driven by the migration behavior of
agents. For example, if we make $F$ bimodal, or if we only have two
types of $\alpha$, in the extremum case, the effect through the migration
is muted. By assuming uniformity, the fundamental effect of the blockchain
platform and market structure are not altered by this slope-effect
(or drastic change in the extensive margin) of migration. This is
the problem that commonly arises in the model of segmented markets
in which heterogeneous traders decide in which one to participate.
See, for example, \citet{zhu2014dark} for a similar discussion.}
\begin{align}
P_{j} & =\begin{cases}
\left(\frac{1}{\tilde{\pi}_{C}}+\frac{1}{\Delta\tilde{\pi}}\right)^{-1}\left(\frac{P_{B}}{\Delta\tilde{\pi}}-K_{C}^{D}\right) & \text{for \ensuremath{j=C}}\\
P_{C}+\Delta\tilde{\pi}(1-K_{B}^{D}) & \text{for \ensuremath{j=B,}}
\end{cases}\label{pc}
\end{align}
with $\Delta\tilde{\pi}\equiv\tilde{\pi}_{B}-\tilde{\pi}_{C}.$ Note
that plugging in the aggregate supply $K_{j}^{S}$---which will be
derived in the next section---yields the equilibrium prices. 

The price in one market affects the price in the other, and the quality
difference, $\Delta\tilde{\pi}$, influences the prices in the following
way:
\begin{lem}
With a fixed supply, a larger quality difference ($\Delta\pi$) induces
more traders to migrate from the $C$-market to the $B$-market, leading
to a lower $P_{C}$ and higher $P_{B}$. 
\end{lem}
From the second equation, the price spread, $\Delta P\equiv P_{B}-P_{C}$,
stems from the quality difference. This can be seen as a premium:
The asset in the $B$-market obtains a higher valuation than the one
in the $C$-market through its higher quality, $\Delta\tilde{\pi}$.
We derive the quality difference from supply-side behavior in the
next subsection.

\subsection{Optimal Behavior of Sellers: Endogenous Quality\label{sub_seller}}

As the literature on adverse selection assumes, each seller knows
the quality of her asset.\footnote{This structure can be generalized by assuming that the seller is informed
with probability $\lambda$ and uninformed with probability $1-\lambda$.
An informed seller knows a specific characteristic of the asset and
can distinguish lemons, while an uninformed agent cannot. We provide
the analyses in the generalized case with $\lambda\in(0,1)$ in Appendix
\ref{sec:Generalized-Model-with}.} Also, note that each seller does not engage in strategic trading:
the signaling effect of venue choice is shunted aside because each
trader is non-atomic. Instead, the sell-side selection (screening)
occurs due to $\theta>0$, even in the competitive equilibrium.\footnote{Notice that, by setting the model in this way, we also implicitly
exclude the possibility of collusion by sellers as in \citet{cong2017blockchain}.
Since sellers are non-atomic, they expect that their behavior does
not affect the market prices, quantity, or quality. } 

\subsubsection{Low-Quality Sellers}

First, we consider the optimal strategy of sellers with low-quality
assets ($L$-type sellers). If a trader sells the asset in the $B$-market,
the expected return is 
\begin{equation}
W_{B}^{L}=(1-\theta)P_{B}+\theta\phi\alpha.\label{eq:WBL}
\end{equation}
The first term represents the case where the transaction avoids rejection,
while the second is the case where the asset is rejected by the blockchain.
In the latter case, the trader must use the asset to get $\phi\alpha$.\footnote{\label{fn:The-alternative-assumption}The alternative assumption is
allowing rejected traders to conduct ``order routing.'' A trader
can first try to sell in the $B$-market and, if rejected, can submit
a sell order in the $C$-market. We can show that this alternative
assumption does not change our main results, including propositions
\ref{prop_pi}, \ref{prop_price}, and \ref{prop_lam1_theta}, though
the equilibrium conditions are slightly modified. The results are
available upon request.} On the other hand, if she sells it in the $C$-market, the return
is $W_{C}^{L}=P_{C}$ , while the return from staying inactive is
$W_{0}^{L}=\phi\alpha.$ See the left-hand panel of Figure \ref{Fig_seller}
for a diagram of these value functions.

\begin{figure}[H]
\begin{center}\caption{Returns for Sellers}\label{Fig_seller}

\includegraphics[scale=0.4]{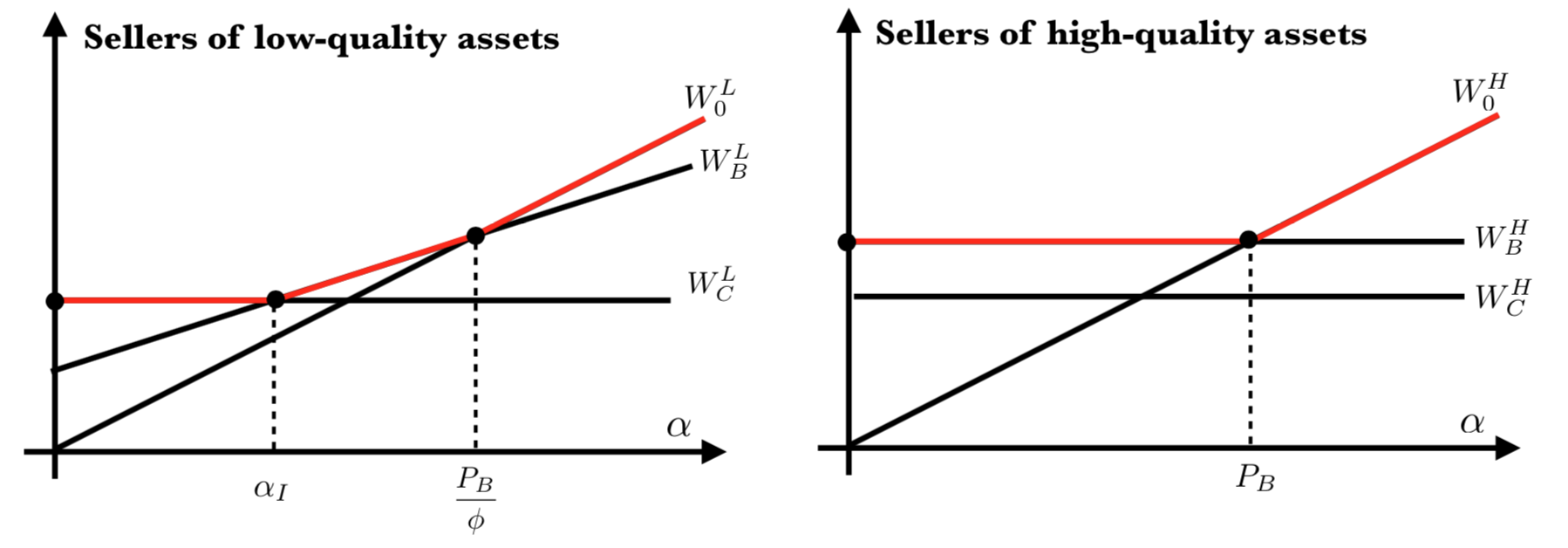}

\end{center}
\end{figure}
By comparing these three returns as functions of $\alpha$ under (\ref{guess}),
we see that the optimal strategy is to (i) stay inactive if $\alpha>\frac{P_{B}}{\phi}$,
(ii) sell in the $B$-market if $\alpha\in(\alpha_{I},\frac{P_{B}}{\phi}]$,
and (iii) sell in the $C$-market if $\alpha\le\alpha_{I}$, where
\begin{equation}
\alpha_{I}=\max\left\{ \frac{P_{C}-(1-\theta)P_{B}}{\phi\theta},0\right\} \label{eq:ThresSeller}
\end{equation}
is the cutoff that separates sellers into the $B$- and $C$-markets. 

When it is strictly positive, the cutoff $\alpha_{I}$ is increasing
in $\theta$, decreasing in $\phi$, and increasing in the expected
price difference (numerator of $\alpha_{I}$). Given the prices, an
increase in $\theta$ makes sellers who traded in the $B$-market
migrate to the $C$-market because a higher rejection probability
lowers their expected profit. On the other hand, a higher $\phi$
increases the continuation value and the profit from selling in the
$B$-market, causing marginal sellers to switch to this platform.
Finally, a larger difference in the expected prices, $P_{C}-(1-\theta)P_{B}$,
makes the $C$-market more attractive. 

High-$\alpha$ sellers are more likely to trade lemons in the $B$-market,
while low-$\alpha$ sellers tend to prefer the $C$-market due to
the price-liquidity tradeoff, \emph{i.e.}, the $B$-market provides
a higher selling price, but at the risk of rejection. High-$\alpha$
sellers do not care about the lower execution probability in the $B$-market
because they can obtain a high value of $\phi\alpha$ even if the
selling order is rejected, while the opposite is true for their low-$\alpha$
counterparts.

\subsubsection{High-Quality Sellers}

For a seller with a high-quality asset ($H$-type seller), the return
from trading in the $C$-market, $B$-market, and not trading are
given by\footnote{More precisely, a seller who sells in the $B$-market obtains $P_{B}/Q$
of cryptocurrency, which amounts to $P_{B}$ in terms of cash value.
We implicitly assume that sellers have access to a dynamic market
for cryptocurrency, in which they can trade it for fiat money at the
same exchange rate $Q$ over time. This assumption is motivated by
the overlapping generations of traders. The structure of these generations
is identical over time, and buyers in their young period arrive at
the markets and demand cryptocurrency as a means of exchange. Conversely,
older sellers are ready to trade their cryptocurrency for fiat money. } 
\begin{align*}
W_{j}^{H} & =\begin{cases}
P_{j} & \text{for \ensuremath{j=B,C,}}\\
\alpha & \text{if \ensuremath{j=0}.}
\end{cases}
\end{align*}
Under the guess (\ref{guess}), the optimal behavior is to (i) stay
inactive if $\alpha>P_{B}$ and (ii) sell it in the $B$-market if
$\alpha\le P_{B}$. The right-hand panel of Figure \ref{Fig_seller}
shows this comparison.

Therefore, the amount of assets that sellers \textit{intend} to sell
in each market is
\begin{align}
S_{B} & =\pi F(P_{B})+(1-\pi)\left[F\left(\frac{P_{B}}{\phi}\right)-F\left(\alpha_{I}\right)\right],\label{SU}\\
S_{C} & =(1-\pi)F\left(\alpha_{I}\right).\label{SI}
\end{align}
In (\ref{SU}), the first term is the supply from $H$-type sellers,
and the second is that from $L$-type sellers. (\ref{SI}) only consists
of $L$-type selling behavior. As suggested by the literature on adverse
selection with segmented markets (\citealp{chen2012market}; \citealp{kim2012endogenous};
\citealp{guerrieri2014dynamic}), a market with a low (high) price
and deeper (shallower) liquidity tends to attract low-quality (high-quality)
assets because the different prices and liquidity can work as a screening
device. 

Since the blockchain technology weeds out $\theta$ fraction of the
lemons from the $B$-market, the supply functions are given by
\begin{equation}
K_{B}^{S}=\pi F(P_{B})+(1-\pi)(1-\theta)\left[F\left(\frac{P_{B}}{\phi}\right)-F\left(\alpha_{I}\right)\right],\label{eq:KBI}
\end{equation}
while, in the $C$-market, all of the selling attempts are accomplished,
$K_{C}^{S}=S_{C}$. As a result, the average quality in each market
is derived as follows:
\begin{lem}
Endogenous market qualities are given by 
\begin{equation}
\pi_{j}=\begin{cases}
\frac{\pi F(P_{B})}{\pi F(P_{B})+(1-\pi)(1-\theta)\left[F\left(\frac{P_{B}}{\phi}\right)-F\left(\alpha_{I}\right)\right]} & \text{if \ensuremath{j=B}}\\
0 & \text{if \ensuremath{j=C.}}
\end{cases}\label{eq:pib}
\end{equation}
\end{lem}
Note that all of the high-quality assets go to the $B$-market since
it provides a better price. In other words, all of the assets traded
in the $C$-market are of low quality. This arises from the information
structure of sellers (\emph{i.e.}, they all know the quality of their
assets). In the real economy, it is not natural to claim that the
$C$-market contains only low-quality assets. Thus, in Appendix \ref{sec:Generalized-Model-with},
we redefine the equilibrium with uninformed sellers to show that a
more general information structure provides $0<\pi_{C}<\pi_{B}<1$
in the equilibrium, while we stick to our current formulation in the
main model to extract clear intuitions. 

\subsection{Equilibrium Spreads}

We can define the general equilibrium as follows:
\begin{defn}
The general equilibrium is defined by the price, quality, and quantity,
($Q,\{P_{j},\pi_{j},K_{j}\}_{j\in\{C,B\}}$), that clear the markets
($K_{j}^{S}=K_{j}^{D}$) with the following equations (under the normalization
of $B_{S}=1$): 
\begin{align}
K_{C}^{S} & =(1-\pi)\alpha_{I},K_{B}^{S}=\pi P_{B}+(1-\pi)(1-\theta)\left(\frac{P_{B}}{\phi}-\alpha_{I}\right),\label{K_lam1}\\
K_{B}^{D} & =1-\frac{P_{B}-P_{C}}{\tilde{\pi}_{B}-\phi},K_{C}^{D}=\frac{P_{B}-P_{C}}{\tilde{\pi}_{B}-\phi}-\frac{P_{C}}{\phi},\nonumber \\
\pi_{B} & =\frac{\pi P_{B}}{K_{B}},Q=P_{B}K_{B}^{i}.\nonumber 
\end{align}
The last equation is the clearing condition for the cryptocurrency
market. 

We can quantify the spreads in the price and quality between the two
markets (see Appendix \ref{app_dp_dpi} for the proofs).
\end{defn}
\begin{prop}
\label{prop_pi}The blockchain market achieves a higher quality than
the cash market, i.e., $\pi_{B}>\pi_{C}$.
\end{prop}
Proposition \ref{prop_pi} has direct implications for the prices.
That is, the positive spread in the quality, $\Delta\pi>0$, results
in a higher price in the $B$-market as well.
\begin{prop}
\label{prop_price}The price of assets traded in the $B$-market is
higher than that in the traditional $C$-market, that is, $P_{B}>P_{C}$.\label{prop:The-price-of}
\end{prop}
This conclusion is consistent with the findings on the wine blockchain
by the EY Advisory (see Section \ref{sec:Empirical-Implications}),
in which wines sold on the blockchain platform attain a higher price.

Note that a higher $\theta$ affects $\pi_{B}$ via two channels.
First, it exogenously precludes $\theta$ fraction of the lemons as
the second term in (\ref{K_lam1}). Second, it generates the endogenous
sorting of low-quality assets which manifests in the change in the
cutoff $\alpha$ caused by the fluctuation of $P_{j}$. This happens
even if the sellers do not trade strategically. Rather, the sell-side
selection is a consequence of the purely competitive tradeoff between
the higher equilibrium price in the $B$-market and detection risk.
Due to its higher continuation values, high-quality assets tend to
cluster in the $B$-market, while low-quality assets cluster in the
$C$-market to avoid being rejected. This mechanism generates a higher
quality and price in the $B$-market (spreads), which are self-sustaining
in the equilibrium.\footnote{The economy is not continuous at $\theta=0$. When we make $\theta\searrow0$,
it converges to the segmented markets economy with two homogeneous
markets, which is different from an economy with only one (the $C$-market).
For this reason, we do not compare the single-market economy with
the segmented economy. Rather, we focus on the comparative statics
in the economy with segmented markets.}

\section{Comparative Statics: The Effect of Security Improvement}

For a technical reason, assume that the following condition holds:\footnote{This condition guarantees that the economy has both cases of $\alpha_{I}>0$
and $\alpha_{I}=0$.} $\pi(1-\pi)(1-\phi)<1/4.$ 

\subsection{Cash-Market Breakdown}

Our first result regarding blockchain security seems rather drastic:
the existence of this platform can completely destroy the activity
of the cash market when the security level $\theta$ is sufficiently\textit{
low}. This appears counterintuitive given the literature on ``money
versus credit,'' which argues that a more sophisticated record-keeping
system makes cash inessential. Our result depends on whether $\theta$
becomes too low to sustain $\alpha_{I}=P_{C}-(1-\theta)P_{B}>0$.\footnote{The term ``cash-market breakdown'' implies that all transactions
that do not go through the blockchain disappear, but it does not intend
to indicate the disappearance of cash due to digital currency with
no blockchain foundation, such as PayPal.} 
\begin{lem}
\label{prop_switch}Let $\theta_{0}$ be the smaller solution of $\theta^{2}(1-\pi)-\theta+\pi(1-\phi)=0$.
(i) $\alpha_{I}$ is positive if and only if $\theta>\theta_{0}$.
(ii) $\theta_{0}$ is decreasing in $\phi$ and $\pi$.
\end{lem}
\begin{proof}
See Appendix \ref{app_switch} for the statement (i). (ii) follows
immediately from the definition of $\theta_{0}$.
\end{proof}
Recall that $\alpha_{I}$ is the cutoff for $L$-type sellers in the
$B$-market ($\alpha\ge\alpha_{I}$) or in the $C$-market ($\alpha<\alpha_{I}$).
Therefore, Lemma \ref{prop_switch} and (\ref{K_lam1}) imply the
following:
\begin{prop}
\label{prop:-market-shuts-down,}The $C$-market shuts down, $K_{C}=0$,
when $\theta$ is smaller than $\theta_{0}$.
\end{prop}
We denote the region $\Theta_{NC}\equiv(0,\theta_{0}]$ as the \textit{no
C-market} region. The result comes from the supply side. Remember
that only $L$-type sellers trade their assets in the $C$-market.
They wish to sell the asset at a higher price $P_{B}$, but they fear
rejection. $\theta<\theta_{0}$ makes the rejection risk sufficiently
small that the price improvement in the $B$-market becomes dominant.
As a result, all the sellers migrate out of the $C$-market and try
to sell in the $B$-market. Of course, this induces a higher $P_{C}$,
but it is bounded and cannot explode due to the existence of outside
options (buying in the $B$-market or staying inactive). Also, the
higher share of $H$ assets in the $B$-market always makes $P_{B}>P_{C}$
even at the limit of $K_{C}^{S}=0$, making the shutdown of the $C$-market
an equilibrium outcome. 

We can determine the situations in which the $C$-market tends to
be eliminated.
\begin{cor}
\label{cor:Smaller--and}A smaller $\phi$ and $\pi$ make the cash-market
breakdown more likely, i.e., they make $\Theta_{NC}$ larger. 
\end{cor}
\begin{proof}
Immediate from Lemma \ref{prop_switch}.
\end{proof}
This result should also be intuitive if we recall that the critical
$\theta$ is expressed as $\theta_{0}=\frac{P_{B}-P_{C}}{P_{B}}$.
The lower $\phi$ and $\pi$ mean that the underlying asymmetric information
is severe, and both make buyers more eager to trade in the $B$-market,
leading to a larger spread, $P_{B}-P_{C}$. Then, sellers of low-quality
assets are more willing to sell in the $B$-market and are likely
to abandon the traditional cash market.\footnote{Technically, we can avoid this by considering the general model with
$\lambda<1$, by restricting our focus on $\theta\ge\theta_{0}$,
and by modifying the model according to the discussion in note \ref{fn:The-alternative-assumption}. }

Of course, it is not realistic to anticipate that the real-world cash
market will completely break down---we do not believe that all grocery
stores will use the blockchain to track information for all products.
However, our discussions can be viewed as the model of an exchange
market for a particular class of assets (\emph{e.g.}, wine, diamonds,
art, and so on), for which the dominance of the blockchain platform
can be more believable. This also has implications for international
trade, as transactions with a country that supplies an ambiguous quality
of goods would be completely done through the blockchain.\footnote{Proposition \ref{prop:-market-shuts-down,} and Corollary \ref{cor:Smaller--and}
indicate that the shutdown occurs when (i) the introduced blockchain
is immature in the sense of $\theta\le\theta_{0}$, and (ii) agents
suffers from severe asymmetric information. If we consider the micro-foundation
of $\theta$ in Appendix \ref{subs:motive}, situation (i) can happen
when the total number of intermediations between producers and buyers
is large and covering all the steps using blockchain transactions
is more difficult, as exemplified by international trade in the real
world. Note that a longer intermediation chain also induces more severe
information asymmetry, which implies that (i) and (ii) may ensure
simultaneously. }

\subsection{Coexistence of Two Markets}

In the following subsections, we let $\theta$ be sufficiently high,
as we are interested in the interaction of two markets.\footnote{We believe that a sufficiently high $\theta$ that sustains the coexisting
$B$- and $C$-markets is realistic given the discussion in the introduction.
Arguments for $\theta\le\theta_{0}$ are provided in Appendix \ref{subsec:Proof-for-Proposition}.} We first investigate how $P_{B}$ and $\pi_{B}$ are differently
affected by $\theta$.
\begin{prop}
\label{prop_lam1_theta}(i) The segmented-markets economy admits a
unique solution in which (ii) $\frac{dP_{B}}{d\theta}>0$, $\frac{d\pi_{B}}{d\theta}>0$,
$\frac{dP_{C}}{d\theta}<0$, and $\frac{dK_{B}}{d\theta}<0$. Moreover,
(iii) the price spread widens more than the quality spread, \emph{i.e.},
$\frac{d\Delta P}{d\theta}>\frac{d\Delta\pi}{d\theta}$. \label{prop:(i)-The-segmented}
\end{prop}
\begin{proof}
See Appendix \ref{app_proof_lam1}.
\end{proof}
Consider an increase in $\theta$. The supply of \textit{high}-quality
assets is not directly affected, as the first term of $K_{B}^{S}$
in (\ref{K_lam1}) suggests. On the other hand, with a fixed $P_{B}$,
a higher $\theta$ reduces $K_{B}^{S}$ by detecting and sweeping
out \textit{\small{}low}-quality assets (the supply-side effect).
Also, this improves the $B$-market's quality, making participants
more willing to buy in $B$-market, which increases $K_{B}^{D}$ (the
demand-side effect). 

How does $P_{B}$ change compared to the quality, $\pi_{B}$? The
small supply and strong demand pressure both work to increase the
price, \emph{i.e.}, it is pushed up by both the demand- and supply-side
effects. On the other hand, the quality improvement is linked only
with the decline in the supply. Crucially, this implies that the growth
in the price is larger than the quality improvement, making the trading
demand in $B$-market smaller. This result has direct implications
for the next results.

\subsection{Non-Monotonic Effects of Blockchain Sophistication}

Our primary concern is whether the sophistication of the blockchain,
measured by $\theta$, increases the transaction value in the platform,
the price of cryptocurrency, the value of the blockchain itself, and
consumers' welfare. It turns out that a rise in $\theta$ has non-monotonic
effects on these variables. Specifically, a higher $\theta$ is more
likely to \textit{\small{}lower} these variables when information
asymmetry is not severe, or the level of $\theta$ is sufficiently
small. 

First, we formally state the results regarding the trading value and
$Q$ in the sequel. We will offer intuitions and key mechanisms behind
the results when we investigate the value of the blockchain and consumers'
welfare in the next subsection, because all of these variables are
driven by the same factors.

\subsubsection{Trading Value and the Price of Cryptocurrency}

The market clearing condition gives $Q=P_{B}\int k_{B}^{D}dF(\alpha)=P_{B}K_{B}$
(with $B_{S}=1$, which makes no difference in our analyses). Let
$\phi_{1}=\frac{2-\pi}{3-\pi}$, and $\phi_{0}(<\phi_{1})$ be the
unique solution of (\ref{app_phi1}) in Appendix \ref{app_proof_lam1}.
\begin{prop}
\label{prop_lam1_Q}\label{prop: Qnonlin}(i) If $\phi<\phi_{0}$,
$Q$ and $P_{B}K_{B}$ are monotonically increasing in $\theta$.
(ii) If $\phi_{0}\le\phi<\phi_{1}$, $Q$ and $P_{B}K_{B}$ are $U$-shaped,
and there is $\theta^{*}$ s.t.,
\[
\frac{dQ}{d\theta}=\frac{dP_{B}K_{B}}{d\theta}\lessgtr0\Leftrightarrow\theta\lessgtr\theta^{*}.
\]
(iii) If $\phi_{1}\le\phi\le1$, $Q$ and $P_{B}K_{B}$ are monotonically
decreasing in $\theta$.\label{prop:(i)-If-,}
\end{prop}
\begin{proof}
See Appendix \ref{app_proof_lam1}.
\end{proof}
The proposition indicates that the trading value and the price of
cryptocurrency exhibit a non-monotonic reaction to the sophistication
of the blockchain technology. 

This finding differentiates our theory from the literature on money
versus credit, in which the increase in record-keeping ability tends
to make cash inessential. As we will see in the next subsection, endogenous
price and quality spreads, which are absent in the literature, are
the key factors that prompt this result.

\subsection{Fundamental Value of the Blockchain and Welfare Impact \label{sec:Welfare-Analyses:-Fundamental}}

In this subsection, we calculate the aggregate welfare of buyers (see
Appendix \ref{subsec_seller_gain} for sellers' welfare), as well
as the welfare gain from access to the blockchain platform.\footnote{Technically, as the CIA constraint always binds, the welfare comparison
does not hinge on the existence of cryptocurrency. More generally,
the equilibrium variables with and without the CIA constraint are
identical except for the formula for $Q$ because of the ``monetary
neutrality'' of cryptocurrency.}

\subsubsection{Buyers' Welfare}

We define the aggregate welfare of buyers by integrating the gain
from trade (we ignore the common constant endowment $w$):
\begin{align}
v_{B} & =\int_{\alpha^{*}}(\tilde{\pi}_{B}\alpha-P_{B})dF+\int_{P_{C}/\phi}^{\alpha^{*}}(\phi\alpha-P_{C})dF\label{vb0}\\
 & =\underset{\propto P_{B}K_{B}=Q}{\underbrace{\int_{\alpha^{*}}(\Delta\tilde{\pi}\alpha-\Delta P)dF}}+\int_{P_{C}/\phi}(\phi\alpha-P_{C})dF\label{vb1}
\end{align}
In equation (\ref{vb0}), the first term is the welfare of buyers
who purchase in the $B$-market, and the second encompasses those
who purchase in the $C$-market. This can be rewritten by using ``welfare
gain'' and ``reservation welfare'' as in (\ref{vb1}). The second
term of (\ref{vb1}) represents the welfare of all the active buyers
from purchasing in the $C$-market, \emph{i.e.}, the reservation welfare
when agents can use only this venue. The first term of (\ref{vb1})
is the gain (increment) in welfare that stems from changing the trading
platform from $C$ to $B$, which only $\alpha\ge\alpha^{*}$ agents
attempt to do. Interestingly, the welfare gain is co-linear with the
trading value in the $B$-market in the benchmark model (see Appendix
\ref{app_welfare_b}). If the blockchain platform has cryptocurrency,
it further implies that the welfare gain is measured by $Q$. The
effect of $\theta$ on $v_{B}$ is analyzed later.

\subsubsection{Fundamental Value of the Blockchain Technology\label{subsec:Fundamental-Value-of}}

The first term of equation (\ref{vb1}) shows that access to the blockchain
technology attains a positive fundamental price. To see this, we introduce
a monopolistic \textit{blockchain manager} who maintains the platform
and determines the level of $\theta$. The existence of this type
of agent is realistic: even though the distributed ledger is managed
by the participants of the network, there is an institution that provides
the exchange platform itself.\footnote{The assumption that the manager is a monopolist is also realistic
at this point, provided that we have a limited number of blockchain
firms for each product. For example, as of February 2018, HyperLedger
is the single leading firm that provides a platform for security trading
by using the blockchain. } 

We discuss the buy-side problem by assuming that consumers must pay
a fee to use the blockchain.\footnote{We discuss on the fee imposed on the sell side in Appendix \ref{subsec:Manager-vs.-Sellers}.}
Let us introduce a pre-trade period $t=-1$ and suppose that a randomly
picked buyer is approached by the blockchain manager (called the ``manager,''
hereafter), who charges a fee $f_{B}$ for access to the $B$-market
before the type $\alpha$ is drawn at $t=0$.\footnote{Another way to think about the price of the blockchain is to conceptualize
$f_{B}$ as contingent on the usage of the $B$-market. In this case,
the profit of buying in the $B$-market is shifted down by $f_{B}$
only if a trader decides to participate. This formulation, however,
generates complicated equilibrium conditions because it changes the
cutoff of each trader. To avoid complications, we focus on a setting
with the \textit{ex-ante} contract.} Note that the behavior of this particular agent does not affect the
expected market result because she is non-atomic. Also, there is no
bargaining, the offer is one shot, and take-it-or-leave-it.

If the buyer declines the contract, her expected welfare stays at
the reservation level in (\ref{vb1}), which we denote as 
\begin{equation}
v_{0}=\int_{\frac{P_{C}}{\phi}}(\phi\alpha-P_{C})dF,\label{eq:v0}
\end{equation}
while access to the blockchain platform provides a welfare (after
the fee) of $v_{B}-f_{B}.$ Thus, the amount of the fee that makes
her indifferent is\footnote{We assume the tie-breaking rule so that an agent accepts the contract
with the blockchain manager if she is indifferent.}
\begin{align}
f_{B} & =\Delta v_{B}\equiv v_{B}-v_{0}=\int_{\alpha^{*}}(\Delta\tilde{\pi}\alpha-\Delta P)dF.\label{gain}
\end{align}
In other words, the blockchain manager can charge a fee up to the
amount of the welfare gain given by (\ref{gain}). We can see this
amount as the ``price'' of the blockchain technology or platform,
since traders are willing to ``buy'' the right to participate in
the $B$-market at the price of $f_{B}$. Moreover, we have the following
intuitive result.
\begin{prop}
The fundamental price of the blockchain is perfectly correlated with
the trading value in the $B$-market and the price of cryptocurrency
$Q$ if the CIA constraint is assumed: \label{prop:The-fundamental-price}
\begin{align}
f_{B} & =\frac{\pi(1-\phi)}{2}K_{B}P_{B}=\frac{\pi(1-\phi)}{2}Q.\label{eq:fee-cryptoprice}
\end{align}
\end{prop}
\begin{proof}
See equation (\ref{vb_Q}) in Appendix \ref{app_welfare_b}.
\end{proof}
\begin{cor}
\label{cor_f}\label{cor:The-sophistication-of}The sophistication
of the blockchain technology has the same impact on $f_{B}$ as proposed
by Proposition \ref{prop_lam1_Q}.
\end{cor}
This proposition suggests that, if the blockchain uses cryptocurrency,
the fundamental value of the technology is perfectly reflected by
the price of cryptocurrency. In other words, the price of both cryptocurrency
and the blockchain entirely depends on how active the transactions
in the $B$-market are, which is measured by the trading value.\footnote{Moreover, the price of the blockchain technology is multiplied by
the coefficient $\pi(1-\phi)/2$. This value is the multiplier of
$\Delta\tilde{\pi}$: when the quality difference is large, the gain
from trading in the $B$-market rather than the $C$-market is high.
When asymmetric information is not severe ($\phi$ is high) or the
economy-wide share of low-quality assets is large ($\pi$ is small),
the price of cryptocurrency magnifies the fundamental value of the
blockchain technology or the welfare gain for buyers (and vice versa). } 

\subsection{Intuitions and Mechanism\label{subsec:Intuitions-and-Mechanism}}

The intuition behind the non-monotonic reaction of $Q=P_{B}K_{B}\propto$$f_{B}$
put forward by Proposition \ref{prop: Qnonlin} and Corollary \ref{cor:The-sophistication-of}
is given by the behavior of $P_{B}$ and $K_{B}$, together with the
migration of buyers.

First, when $\theta$ increases, we find that a more secure blockchain
technology tends to widen both price and quality spreads, $\Delta P$
and $\Delta\pi$. The former reduces the demand in the $B$-market,
while the latter increases it, \emph{i.e.}, the $B$-market guarantees
a higher quality but becomes exclusive. Second, the formula $Q=P_{B}K_{B}\propto f_{B}$
implies that $Q$ and $f_{B}$ rise when $P_{B}$ increases more than
$K_{B}$ declines. Rewriting the derivative of $Q$ by using elasticity
makes this clearer. Since $K_{B}$ can be expressed as a function
of $P_{B}$ (without $\theta$), and $P_{B}$ is monotonic regarding
$\theta$, we have
\[
\frac{df_{B}}{d\theta}\propto\frac{dQ}{d\theta}=(1-\varepsilon_{PK})K_{B}\frac{dP_{B}}{d\theta}
\]
with 
\[
\varepsilon_{PK}\equiv-\frac{dK_{B}/dP_{B}}{K_{B}/P_{B}}.
\]
$\varepsilon_{PK}$ is the price elasticity of the $B$-market transaction
volume. Thus, if the price elasticity of demand is high, a decline
in $K_{B}$ dominates the increase in $P_{B}$, leading to a smaller
$Q$ and $f_{B}$. To understand the determinants of $\varepsilon_{PK}$,
recall that the buyers' venue choice is driven by how easily they
can migrate to the $C$-market to avoid a higher $P_{B}$.

When $\phi$ is sufficiently large, asymmetric information is not
severe because the difference between the two asset types is small.
Then, buyers are not eager to have $H$-type assets and are not attracted
to a high $\pi_{B}$ in the $B$-market. Thus, a marginal increase
in $P_{B}$ leads to a larger decline in $K_{B}$, and the transaction
activity in the $B$-market, measured by the transaction value, $K_{B}P_{B}$,
diminishes. Hence, the price of the blockchain platform and cryptocurrency
drops. If $\phi$ is small, it becomes difficult for consumers to
migrate away to the $C$-market, leading to an increase in $P_{B}K_{B}$,
$Q$, and $f_{B}$.

If $\phi$ is intermediate, the level of $\theta$ matters because
it determines the difference between the two markets, $\Delta\pi$.
If $\theta$ is small, so is $\Delta\pi$: the difference in buying
in the $B$-market and $C$-markets is not significant in terms of
the probability of purchasing low-quality assets. This facilitates
migration to the $C$-market, since this market provides a lower price,
while the difference in quality is negligible. This leads to a decline
in $K_{B}$ more than an increase in $P_{B}$, lowering $Q$ and $f_{B}$.
If $\theta$ is large, the $B$-market provides a significantly higher
average quality, \emph{i.e.}, the quality spread is large, and $Q$
and $f_{B}$ increase with $\theta$. 

The bottom line is that, depending on the underlying information asymmetry,
the change in the market structure has a different impact on the market
activity. Specifically, even if the blockchain technology could reduce
asymmetric information, it does not always make this market attractive
for consumers and may even dampen its trading value. 

\subsection{Optimal $\theta$ and Welfare Distortion}

Now, we seek to determine the optimal level of $\theta$ from the
perspective of traders' welfare and the blockchain manager. Suppose
that the manager tries to maximize her fee revenue from the buy side
of the market, $f_{B}$.\footnote{The previous subsection assumes that only one buyer is offered the
contract. Even if the entire set of buyers is offered it, maximizing
$f$ is still optimal since the measure of buyers is one and they
are \textit{ex-ante} identical. } The analogous discussion on fee maximization when it is imposed on
sellers is provided in Appendix \ref{sec:Figures-and-Tables}. In
this subsection, we compare the maximization of the fee by the manager
to the maximization of buyers' aggregate welfare, which may be performed
by a social planner, \emph{e.g.}, FinTech regulation (or promotion)
by the government. Note that the choice of the objective function
is highly arbitrary. However, the evidence from the wine blockchain
by the EY Advisory suggests that the platform imposes a fee on the
buy side of the market.

First, $\theta$ has the following impact on the aggregate consumers'
welfare $v_{B}$:
\begin{prop}
\label{prop_welfare_b}(i) $\frac{dv_{0}}{d\theta}>0$. (ii-1) When
$\pi>1/2,$ $\frac{dv_{B}}{d\theta}>0$.\\
(ii-2) When $\pi\le1/2,$ there is a unique $\phi_{2}$. If $\phi<\phi_{2}$,
then $\frac{dv_{B}}{d\theta}>0$. Otherwise, there is a unique $\theta^{**}\in(0,1]$
such that 
\[
\frac{dv_{B}}{d\theta}\gtrless0\Leftrightarrow\theta\gtrless\theta^{**}.
\]
\end{prop}
\begin{proof}
See Appendix \ref{app_welfare_b}.
\end{proof}
Together with $Q$ and $f_{B}$, buyers' welfare also has a $U$-shaped
trajectory for a certain set of parameters. The reservation welfare
is monotonically increasing in $\theta$ because a higher $\theta$
lowers $P_{C}$ more than it decreases $\pi_{C}$ due to the same
mechanism as in Proposition \ref{prop:(i)-The-segmented}-(iii). 

The remaining part of $v_{B}$, which perfectly correlates with $P_{B}K_{B}$,
generates non-monotonicity in $v_{B}$ by the same mechanism as mentioned
in Subsection \ref{subsec:Intuitions-and-Mechanism}. 

Moreover, the result depends on $\pi$. When $\pi$ is relatively
high, the marginal increase in the fraction of assets rejected by
the blockchain, $(1-\pi)\theta$, is small. That is, innovation does
not cause a large quality improvement or a huge reduction of $K_{B}^{S}$
since the economy does not have a significant amount of low-quality
assets to begin with. The increment in $P_{B}$ caused by the higher
$\theta$ is not significant enough to confound the demand in the
$B$-market, and the welfare gain represented by the first term in
(\ref{vb1}) stays high.

\subsubsection{Optimal $\theta$ for the Platform Manager}

By looking at (\ref{vb1}) and (\ref{gain}), we notice that the objective
functions of the blockchain manager and the social planner are different,
as the manager does not care about the reservation welfare, $v_{0}$.
From (\ref{eq:v0}), we also know that a higher $\theta$ monotonically
increases $v_{0}$ by lowering the price in the $C$-market. Thus,
the manager\textit{ undervalues} the marginal benefit of increasing
$\theta$ compared to buyers' total welfare. 

Formally, let $\theta_{M}^{*}=\arg\max_{\theta\in[\theta_{0},1]}f_{B}(\theta)$
and $\theta_{V}^{*}=\arg\max_{\theta\in[\theta_{0},1]}v_{B}(\theta)$,
which represent the levels of $\theta$ that maximize the fee and
buyers welfare, respectively. Even though it is difficult to analytically
determine $v_{B}(\theta=1)\gtrless v_{B}(\theta=\theta_{0})$, it
is obvious that $\theta_{M}^{*}\neq\theta_{V}^{*}$ when $f_{B}$
is monotonically decreasing and $v_{B}$ is monotonically increasing. 
\begin{prop}
\label{Prop_opt_B}If \{$\pi>\frac{1}{2}$ and $\phi\in[\phi_{1},1]$\}
or \{$\pi\le\frac{1}{2}$ and $\phi\in[\phi_{1},\phi_{2}]$\}, then
$\theta_{0}=\theta_{M}^{*}<\theta_{V}^{*}=1$. If \{$\phi<\phi_{0}$
and $\pi>\frac{1}{2}$\} or \{$\pi\le\frac{1}{2}$ and $\phi<\phi_{2}$\},
then $\theta_{M}^{*}=\theta_{V}^{*}=1.$
\end{prop}
Proposition \ref{Prop_opt_B} tells us that, depending on the parameters,
welfare loss arises from the conflicting objectives of the manager
and the government. The numerical results are shown in Figure \ref{Fig_opt}
when $f_{B}$ or $v_{B}$ has a $U$-shaped curve. 

\begin{figure}[H]
\begin{center}\caption{Fee Revenue and Buyers' Welfare}\label{Fig_opt}

\includegraphics[scale=0.25]{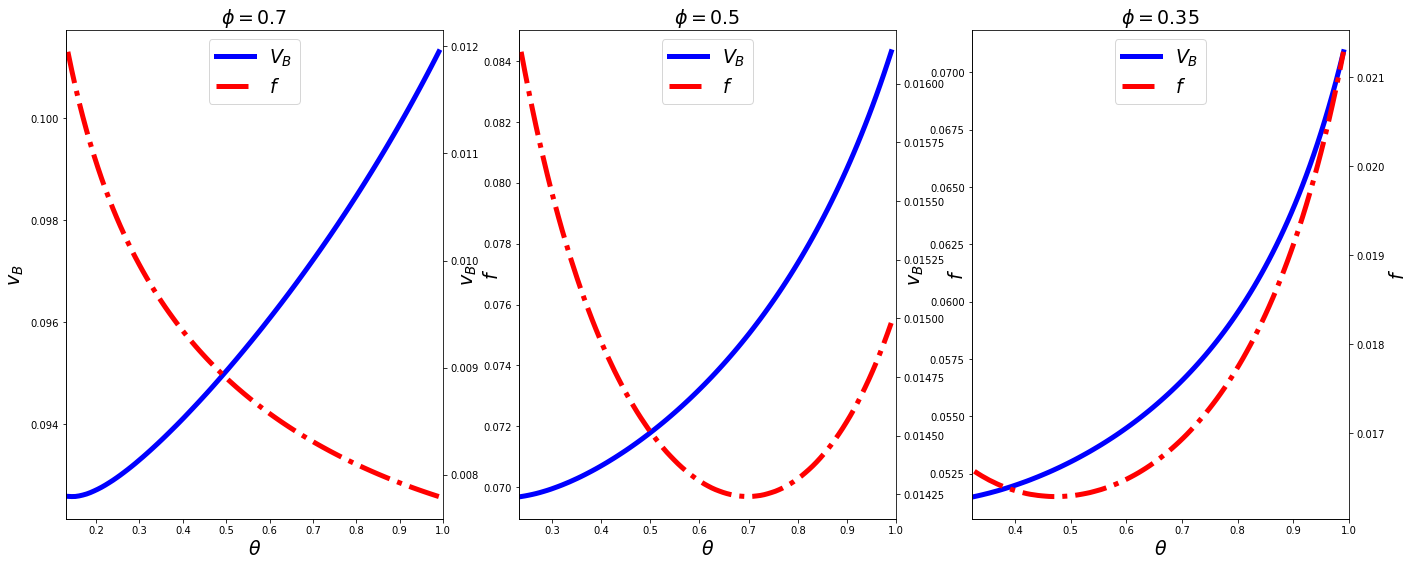}

\end{center}
\end{figure}
This highlights an interesting implication. If the underlying asymmetric
information is mild ($\phi$ is high), as in the left and middle panels,
the marginal increase in $\theta$ tends to dampen the activity in
the $B$-market. This results in a lower welfare gain in the $B$-market
and reduces the fee revenue for the manager. Thus, the manager prefers
to keep $\theta$ low ($\theta=\theta_{0}$). The social planner,
however, knows that a higher $\theta$ boosts the reservation welfare,
and this increment can compensate for losses in the $B$-market when
$\phi$ is relatively high. The level of $\theta$ that maximizes
buyers' welfare is therefore $\theta=1$. Thus, the blockchain platform
operated by the manager is under-secured\textit{ }in the sense that
the reduction of asymmetric information is not enough to achieve the
maximum $v_{B}$.

On the other hand, when asymmetric information is relatively severe,
as in the right panel, a higher $\theta$ facilitates activity in
the $B$-market because $P_{B}$ increases more than $K_{B}$ declines.
In this case, the fee revenue positively responds to a higher $\theta$,
and so does $v_{B}$. Therefore, the blockchain market operated by
the manager can maximize buyers' welfare.

The literature on strategic management, such as \citet{teece1986profiting}
and \citet{brandenburger1996value}, suggests that a firm does not
fully adopt innovation, although it creates value for consumers. This
is because a firm cannot extract full welfare gain of consumers generated
by innovation. We show that this issue arises in the blockchain technology
as well, since the manager cannot extract the value in the traditional
$C$-market created by the blockchain technology.

\subsubsection{Government Intervention}

The preceding discussion indicates that the blockchain manager values
an increment in $\theta$ as highly as the social planner only if
the price elasticity of the demand is small, \emph{i.e.}, a higher
$\theta$ boosts the trading value in the $B$-market. This coincidence
tends to occur when the underlying information problem is severe because
it imposes a higher cost on the migration of buyers. If the market
is closer to complete in terms of $\phi$ or $\Delta\pi$, the manager
prefers a lower $\theta$ than the socially optimal level since she
dislikes a decline in the transaction value in the $B$-market that
is caused by the small cost of changing trading platforms.

This implies that the government should intervene in the intermediation
chain to facilitate transactions through the blockchain and to increase
$\theta$ when the traded goods suffer from non-severe asymmetric
information. In contrast, it should remain neutral when the information
problem is severe since the manager voluntarily maximizes consumers'
welfare. This runs counter to the traditional views on government
intervention in markets with adverse selection (\emph{e.g.}, OTC markets
after the recent financial crisis), which believe the government should
meddle when adverse selection is more severe to avoid market breakdowns. 

Our conclusion is driven by the fact that asymmetric information arises
among agents, while the platform manager, who has a tool to mitigate
the problem, is interested only in the fee revenue from a certain
part of the market. The government does not have a tool to detect
the lemons and must rely on the technological innovation in our model
that may different from the situation in which it intervened in OTC
markets.

\section{Empirical Implications\label{sec:Empirical-Implications}}

We can derive several empirical inferences regarding the fundamental
value of cryptocurrency and the blockchain and their comparative statics.
If we take the model with a CIA constraint, we get the following arguments.
(i) The asset price in the blockchain platform is higher than that
in the cash market (Proposition \ref{prop:The-price-of}). (ii) As
the blockchain system becomes secure, the asset price in the blockchain
platform (resp. cash market) increases (resp. decreases) as in Proposition
\ref{prop:(i)-The-segmented}. (iii) For the same situation, the trading
value in the blockchain platform and the cryptocurrency price increase
if the asymmetric information is severe, while it declines otherwise
(Proposition \ref{prop:(i)-If-,}). 

If we have a dataset that contains the transaction price in a market
with the blockchain technology, smart contracts, and cryptocurrency,
we can test implication (i) by comparing this price to that in the
traditional market. In addition, if we have data from the scratch
of the transaction system, we can keep track of the price in the blockchain
market and the corresponding price in the traditional market to verify
implication (ii). 

Implication (iii) is striking: improvements in the blockchain security
system do not necessarily increase the trading value and demand for
cryptocurrency. On the one hand, this implies that enhancement in
blockchain security does not have a robust testable implication. On
the other hand, with a dataset and a sufficient exogenous change in
$\theta$, our model provides a new measure for the degree of asymmetric
information and adverse selection by analyzing how $\theta$ affects
the transaction value in the $B$-market.

Additionally, considering the welfare results in Subsection \ref{sec:Welfare-Analyses:-Fundamental},
the value of the blockchain system is proportional to the fundamental
price of cryptocurrency (Proposition \ref{prop:The-fundamental-price}).
This has several applications. First, if we have data that measure
the value of the blockchain defined in Subsection \ref{subsec:Fundamental-Value-of}
(\emph{e.g.} \textit{ex-ante} entrance fee in the $B$-market) and
the cryptocurrency price therein, we can directly test the implications
of (\ref{eq:fee-cryptoprice}). Second, even if transactions are not
done using cryptocurrency, Proposition \ref{prop:The-fundamental-price}
tells us how we can predict the welfare-relevant performance of the
blockchain platform. 

Since, as of the date of this study, the application of the blockchain
in state-contingent transactions is still in its initial stages, we
will empirically evaluate these implications in future projects. As
the first step, a qualitative finding consistent with our theoretical
model arises from the introduction of the blockchain in the wine supply
chain. A consulting firm, EY Advisory \& Consulting Co. Ltd. (EY),
pioneered the blockchain-based administration of the quality of each
step of the production of wine, such as grape harvesting, fermentation
and bottling, wholesaling, and retail. From talks with corresponding
consultants, we confirm that the aim of introducing the technology
is to enhance the satisfaction of both customers and suppliers, by
guaranteeing the products quality.

Information on the wine blockchain was kindly shared by EY Japan.
It reveals the financial results of two clients in 2018. One client's
retail price per bottle increased from 7.00 to 9.20 Euros, whereas
other's increased from 7.00 to 7.46 Euros. Under the assumption that
the underlying trend of wine prices is constant, this finding is consistent
with empirical prediction (i). The report also contains information
on the investment and ROI, which ware 53,000 Euros and 7.92\%, and
113,000 Euros and 13.94\%, respectively. These numbers do not include
the value of improving business efficacy due to the blockchain, such
as digitalization and more efficient management. In sum, investment
in the blockchain generates positive return for the client firms.

\section{Conclusion}

We develop a simple model to analyze some economic implications of
the blockchain technology as a new transaction platform. Following
the notion of the smart contract, we consider the blockchain protocol
as a way to mitigate asymmetric information and investigate the effect
of technological sophistication (innovation) on the economy.

Firstly, the blockchain as a platform causes the segmentation of the
trading venues and the differentiation of both sides of the markets
(buyers and sellers). We consider asymmetric information among the
agents and show that the segmentation and differentiation endogenously
generate spreads in the price and quality of the assets traded in
segmented markets. 

We find that the sophistication and innovation of the blockchain have
non-monotonic effects on the trading value in the blockchain, the
fundamental value of the platform, the price of cryptocurrency, and
consumers' welfare. That is, innovation does not necessarily increase
the value of the blockchain and consumers' welfare. This is because
a more sophisticated blockchain attracts high-quality assets and boosts
their price. Since the price increases more than the quality, the
blockchain platform becomes ``an exclusive market.'' When the underlying
asymmetric information is not severe, innovation makes a large number
of consumers migrate away from the blockchain platform to the traditional
one, because they are willing to accept lower-quality assets to save
a price cost.

The non-monotonicity leads to a welfare loss when a platform manager,
who competes with the traditional market, controls the level of innovation.
Since a very sophisticated blockchain platform is not attractive for
most consumers, the platform cannot charge a high access fee. Thus,
the manager has an incentive to keep the innovation level lower than
the first best.

A few issues, such as empirical implications, cannot be investigated
well without the availability of further data. Nonetheless, this model
proposes the first theoretical framework to study the measurable outcomes
of new digital innovations. In addition, one possible future project
is the extension of this framework into a dynamic setup. Specifically,
this model can be modified to analyze a structure of overlapping generations
and time-varying stochastic dividends of the assets, as per the previous
literature. Together with the blockchain mechanism, the supply function
of cryptocurrency is another salient difference of the blockchain
from traditional cash, as \citet{schilling2018some} point out, and
incorporating both of these factors provides a more comprehensive
pricing theory for cryptocurrency. 

Even though the blockchain technology and cryptocurrency are still
in their nascent and pivot around speculation, their influence is
growing and their potential applications are being vigorously sought.
Therefore, we believe that the analyses of their fundamental effects
in our theoretical model will have important implications not only
for financial markets, but also for the entire economy.

\setstretch{1}\bibliographystyle{aer}
\bibliography{/Users/daisukeadachi/Dropbox/projects/Blokchain_Bitcoin/Papers/blockchain}
\setstretch{1}

\appendix

\section{Appendix: Motivating Examples for $\theta$\label{subs:motive}\label{sec:Appendix:-Motivating-Examples}}

\subsubsection*{Smart Contract: Reduction in Asymmetric Information}

Consider an agent who wants to purchase a good (say, a box of wines).
A value of the wine for the consumer depends on multiple dimensions
of state, $S=(s_{1},s_{2},\cdots,s_{N})$. We can think of them as
a brand of grape, a producer, vintage, storage conditions, and so
on. There are $N$-steps of intermediations between a wine producer
and consumer, and each step is operated by an anonymous intermediary
whose type is either $H$ or $L$ (see Figure \ref{Fig_IB}). The
state $s_{j}$ denotes the type of the intermediary and takes two
values, $s_{j}\in\{s^{H},s^{L}\}$, with $\Pr(s_{j}=s^{H})=p$. 

We simplify the arguments by assuming that the consumer's private
value of the good is positive if, and only if, all the states are
high, $S=S^{H}\equiv(s_{1}^{H},s_{2}^{H},\cdots,s_{N}^{H})$. Otherwise,
the good is valueless. Each intermediary is rewarded equally only
if the good is sold.\footnote{Rewards do not have to be specified in this example. Any positive
rewards, such as private value and monetary payoff, contingent on
the purchase of goods by buyers generate the same results.}

To describe asymmetric information about the quality of the good,
suppose that a ``label'' of the wine tells only an announced states
$\hat{s}$ alleged by intermediaries, and the true state is not verifiable:
the consumer gets to see only $\hat{S}=(\hat{s}_{1},\cdots,\hat{s}_{N})$.
Since the consumer's private value is positive only if $S=S^{H}$,
announcing $\hat{s}_{j}=s^{H}$ is optimal for all $j$, which results
in $\hat{S}=S^{H}$.\footnote{If the good contains $s^{L}$ for some step-$j$, the consumer does
not buy the product from the intermediary-$N$ (\emph{i.e.}, a retailer).
Then, if $\hat{s}_{j}=s^{H}$, $\forall j\le N-1$, it is optimal
for the retailer to announce $s^{H}$ and sell it at her store. On
the other hand, if there is some $j\le N-1$ who announced $s^{L}$,
then the retailer does not accept the goods knowing that she cannot
sell them to the consumer. By taking this argument backward, we can
say that the all goods sold by the retailer have the same label with
$\hat{S}=S^{H}$.} This describes a typical situation in which a consumer is devoid
of a comprehensive knowledge to value a good---it is hard to identify
the quality of a wine before she purchases and drinks it.\footnote{Adopting these arguments into Bitcoin blockchain is easy; The traded
asset is bitcoin itself, state $s_{j}$ represents the balance of
bitcoin on traders' accounts at date $t=j$, and traders may have
a transaction or liquidity shock (state) in each period, which determines
the state in the current period. For instance, $s_{t}$ is either
``spent $x$ amount of coin ($s_{t}=s_{t-1}-x$)'' or ``earn additional
$y$ amount of coin ($s_{t}=s_{t-1}+y$)'' with some probability.} This is represented by the first row of intermediations in Figure
\ref{Fig_IB}. 
\begin{figure}[H]
\begin{center}\caption{Intermediation by Cash and Blockchain}
\label{Fig_IB}

\includegraphics[scale=0.5]{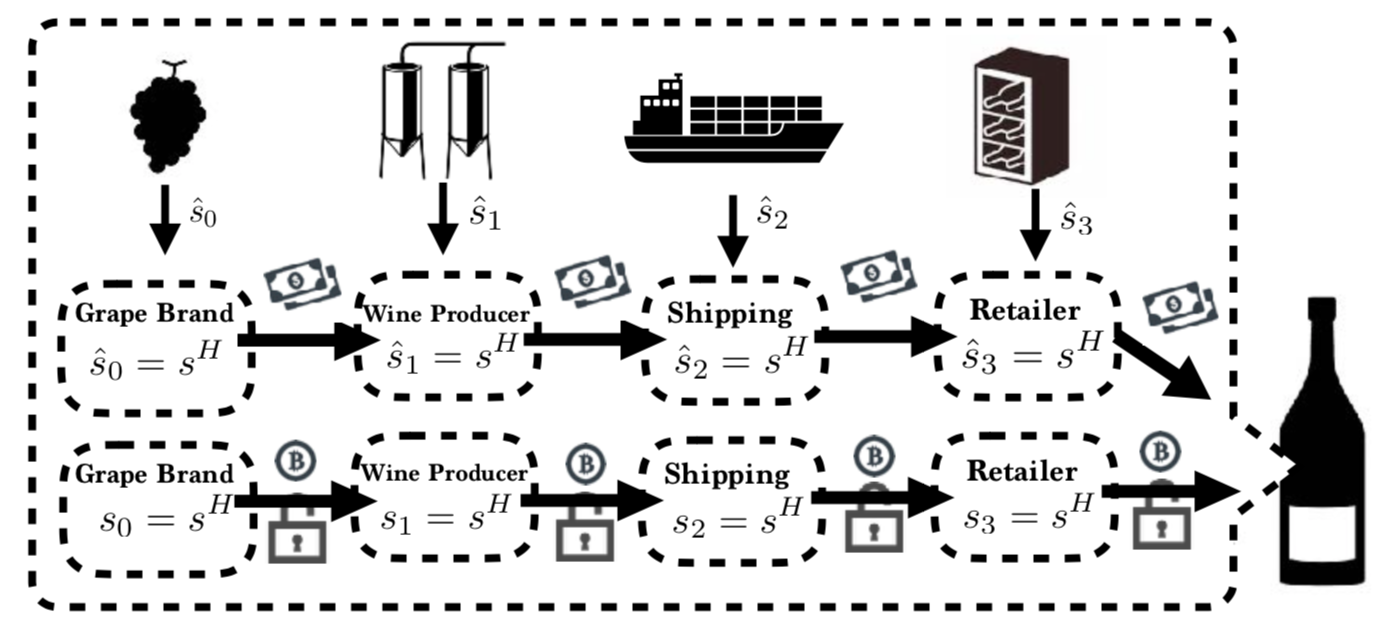}

\end{center}
\end{figure}

Now, we introduce the blockchain protocol (the second row of intermediations
in Figure \ref{Fig_IB}). If the transaction at step $j$ is consummated
through the blockchain, then the announced state $\hat{s}_{j}$ is
supposed to be credible, \emph{i.e.}, $\hat{s}_{j}=s_{j}$, and the
scripts on Ethereum make transactions take place only if all the past
states are $H$. 

We define $\hat{\theta}$ as the \textit{fraction of intermediations
that adopt blockchain transaction}, and $\theta:[0,1]\rightarrow[0,1]$
as the\textit{ fraction of low-quality goods rejected in the intermediations
chain} (as a function of $\hat{\theta}$). Define $\pi(\hat{\theta})\equiv\Pr(S=S^{H}|\hat{\theta})$.\footnote{A set of information also includes $\hat{S}$, but this does not convey
any information since all of the intermediaries have an incentive
to announce the high state regardless of their true types.} 

If there is no transaction conducted by the blockchain, we have $\pi(0)=p^{N}\equiv\pi$,
and $1-\pi$ fraction of the goods in the retail store (step $N$)
are of low-quality. We see this as a benchmark, in which only traditional
means of transaction is used. Now, suppose that $0<M\le N$ steps
of intermediations are conducted by the blockchain technology; $\hat{\theta}=\frac{M}{N}$.\footnote{Since we set $p_{j}=p$, we can assume, without loss of generality,
that first $M$ steps are executed through blockchain.} Then, the resulting probability of the high-quality goods in the
retail store is $\pi(\hat{\theta})=p^{N-M}.$ By the definition of
$\theta$, this can be expressed in terms of $\theta$:
\[
\pi(\hat{\theta})=\frac{\pi}{\pi+(1-\theta)(1-\pi)},
\]
which is the probability of the high-quality goods conditional on
the goods are sold in the retail store after $\theta$ fraction of
low-quality goods are rejected. By equating these two expressions,
we can rewrite $\theta$ as a monotonically increasing function of
$\hat{\theta}$: 
\[
\theta=\frac{1-p^{N\hat{\theta}}}{1-p^{N}}\in[0,1].
\]
Thus, in this discussion, as the number of transactions founded on
blockchain technology ($\hat{\theta}$) increases, the probability
of rejecting low-quality goods ($\theta$) monotonically increases. 

In the main model, we use $\theta$ as a metric of the blockchain
security, \emph{i.e.}, the power of the blockchain technology to reduce
the economy wide asymmetric information. In other words, we can think
of this example as a micro-foundation of $\theta$ in the main model
by making the step-$N$ retailer the ``sellers.'' In Appendix \ref{sec:Technology-Overview:-Cryptocurre-1},
we adopt this architecture into some real-world examples: Bitcoin
and Ethereum. It also provides examples of blockchain platforms not
associated with circulation of cryptocurrency.

\subsubsection*{Time-Consuming Transactions}

One of the most salient benefits of state-contingent transactions
manifests itself in international trade or remittance. It is well
known that it takes a huge cost and time to settle international trade
of goods because it involves mostly manual paper works, authorization
of banks in both countries, and jurisdiction problems. This also applies
to the international remittance in which we need to authenticate bank
accounts of both parties. 

We can describe this situation by using the example above. Suppose
that Alice in California wishes to send money to Bob in Africa, while
there is a chance that Bob's account is not authentic and he may run
away without sending back money or goods. The start of the chain ($j=0$)
is Bob who is either a good or bad agent (bank account is authentic
or not), and all other intermediaries ($j=2\sim N$, say banks) try
to verify that Bob at $j=0$ is authentic. 

For banks, verification may take a long time or even impossible ($s_{j}=L$)
with probability $1-p$. Alice is in need of immediacy, and $S\neq S^{H}$
takes a toll due to a delay cost. If the transaction is conducted
by Bitcoin, however, it can be dramatically secure because of the
above-mentioned mechanism and no longer takes a long time (it makes
$s_{j}=s_{H}$). We can interpret $\hat{\theta}$ as a fraction of
transactions that introduce the blockchain and reduce the time for
verification. Then, it reduces the possibility of delay by $\theta$,
making the trade more efficient. 

\newpage{}

\begin{center}\begin{Huge}Online Appendix\end{Huge}\end{center}\setcounter{page}{1}

\section{Online Appendix: Generalized Model with Uninformed Sellers\label{sec:Generalized-Model-with}}

Consider the same structure as in the main model. In addition, assume
that a seller is informed with probability $\lambda$ and uninformed
with probability $1-\lambda$. If one is informed, she knows a specific
characteristic of high-quality assets and can distinguish the lemons,
while uninformed agents cannot tell the difference.\footnote{In this case, assume, for simplicity, that the realization of $\alpha$
is independent of the realization of being informed or uninformed.} The optimal behavior of informed sellers is same as the main model.

\subsection{Optimal$\ $Behavior$\ $of$\ $Uninformed$\ $Sellers}

Behavior of an uninformed seller is determined by comparing the following
returns:
\begin{align}
W_{0}^{U} & =(\pi+\phi(1-\pi))\alpha,\nonumber \\
W_{C}^{U} & =P_{C},\nonumber \\
W_{B}^{U} & =(\pi+(1-\pi)(1-\theta))P_{B}+(1-\pi)\theta\phi\alpha.\label{WBU}
\end{align}
The first one is the return from consuming her own asset, the second
one is the return from selling in the $C$-market, and the last one
is the return from selling in the $B$-market. In the last case, she
obtains $P_{B}$ if the transaction is completed, while she ends up
consuming her asset if her order is rejected. The two coefficients
in (\ref{WBU}) represent the probability of successful trade and
rejection. Let 
\[
\tilde{\pi}\equiv\pi+\phi(1-\phi),\pi_{0}\equiv\pi+(1-\pi)(1-\theta)
\]
and define a parameter
\[
\xi\equiv\frac{\pi+(1-\pi)(1-\theta)}{\pi+\phi(1-\pi)(1-\theta)}\tilde{\pi}.
\]

The behavior of uninformed sellers is similar to those of informed
sellers with low-quality assets since both of them fear the risk of
detection. As we can see from (\ref{WBU}), however, the return from
selling in the $B$-market, $W_{B}^{U}$, is lower than that of informed
sellers, $W_{B}$ in (\ref{eq:WBL}), because the return is discounted
by the probability that her asset is of low-quality. On the other
hand, the return from selling in the $C$-market is not affected by
this. Namely, with 100\% probability, they can sell the asset of unknown
quality. As a consequences, once again, it becomes a price-liquidity
tradeoff given the expected continuation value of the asset, which
makes relatively low(resp. high)-$\alpha$ sellers trade assets in
the $C$-market (resp. $B$-market).

\subsubsection*{Sufficiently low price in the $B$-market}

Therefore, if the price in $B$-market is sufficiently low ($\xi P_{B}\le P_{C}$),
trading in the $B$-market is not optimal: they try to sell in the
$C$-market or stay inactive. Hence, there is a unique threshold 
\[
\alpha^{U}=\frac{P_{c}}{\tilde{\pi}}.
\]
This separates sellers who go to the $C$-market and stay inactive.
The amount of sell orders from \textit{uninformed} sellers is
\begin{align*}
S_{B}^{U} & =0,\\
S_{C}^{U} & =(1-\lambda)F\left(\frac{P_{C}}{\tilde{\pi}}\right),
\end{align*}
and it directly corresponds to the supply amount: $K_{j}^{U}=S_{j}^{U}$.

\subsubsection*{Sufficiently high price in the $B$-market}

On the other hand, if the price in the $B$-market is sufficiently
high, $\xi P_{B}>P_{C}$, the uninformed sellers use both of the two
markets because the higher price in the $B$-market strictly outweighs
the risk of holding lemons for high-$\alpha$ sellers. That is, there
are two thresholds,
\[
\alpha_{0}^{U}=\frac{P_{C}-\pi_{0}P_{B}}{\phi\theta(1-\pi)},\alpha_{1}^{U}=\frac{\pi_{0}}{\pi+\phi(1-\pi)(1-\theta)}P_{B},
\]
which separate uninformed sellers into three groups. As in the case
of informed sellers of low-quality assets, uninformed sellers (i)
sell in the $C$-market if $\alpha\le\alpha_{0}^{U}$, (ii) in the
$B$-market if $\alpha\in(\alpha_{0}^{U},\alpha_{1}^{U}]$, and (iii)
stay inactive otherwise. Hence, the amount of sell orders from uninformed
traders is 
\begin{align*}
S_{B}^{U} & =(1-\lambda)[F(\alpha_{1}^{U})-F(\alpha_{0}^{U})],\\
S_{C}^{U} & =(1-\lambda)F(\alpha_{0}^{U}),
\end{align*}
and the supply after the screening by the blockchain is
\begin{align*}
K_{B}^{U} & =(1-\lambda)\pi_{0}[F(\alpha_{1}^{U})-F(\alpha_{0}^{U})],\\
K_{C}^{U} & =(1-\lambda)F(\alpha_{0}^{U}).
\end{align*}

\subsection{Aggregate Supply and Market Quality}

The supply functions in the previous subsections determine the aggregate
supply, $K_{B}^{S}$ and $K_{C}^{S}$, as well as the market quality,
$\pi_{B}$ and $\pi_{C}$. Let $\chi$ be an indicator function for
$\xi P_{B}>P_{C}$, i.e., $\chi=\mathbb{I}_{\{\xi P_{B}>P_{C}\}}$.
The aggregate supply sums up the supply from both types of sellers:
\begin{align}
K_{C}^{S} & =\lambda(1-\pi)F\left(\alpha_{I}\right)+(1-\lambda)\left[\chi F(\alpha_{0}^{U})+(1-\chi)F\left(\frac{P_{C}}{\tilde{\pi}}\right)\right]\label{KCS}\\
K_{B}^{S} & =\lambda\left\{ \pi F(P_{B})+(1-\pi)(1-\theta)\left[F\left(\frac{P_{B}}{\phi}\right)-F\left(\alpha_{I}\right)\right]\right\} \label{KBS}\\
 & +(1-\lambda)\pi_{0}\chi[F(\alpha_{1}^{U})-F(\alpha_{0}^{U})].\nonumber 
\end{align}
By using these equations, we can derive the average quality in both
markets:
\begin{align}
\pi_{C} & =\frac{(1-\lambda)\pi\left[\chi F(\alpha_{0}^{U})+(1-\chi)F\left(\frac{P_{C}}{\tilde{\pi}}\right)\right]}{K_{C}^{S}}\label{pic}\\
\pi_{B} & =\frac{\lambda\pi F(P_{B})+(1-\lambda)\pi\chi[F(\alpha_{1}^{U})-F(\alpha_{0}^{U})]}{K_{B}^{S}}.\label{pib}
\end{align}
The determination of $Q$ is the same as before.

\subsection{Numerical Examples for the General Model}

Figure \ref{Fig_bigphi} plots the effect of $\theta$ on the economic
variables when asymmetric information is not severe ($\phi=0.7)$.\footnote{Parameter values for the numerical examples are given by $\lambda=1$
and $\pi=0.3$. } As we have anticipated, the improvement of the blockchain security
brings about the higher price $P_{B}$ and quality $\pi_{B}$ in the
$B$-market. However, the direct rejection of $\theta$ fraction of
low-quality assets, as well as the higher price, will have a negative
effect on the total trading volume in $B$-market and $Q$. The intuition
is the same as in the main model proposed in subsection \ref{subsec:Intuitions-and-Mechanism}. 

As asymmetric information becomes more severe ($\phi=0.5$), it becomes
more costly for buyers to switch to the $C$-market. Figure \ref{Fig_middlephi}
provides effects of improvement in the blockchain technology. When
$\theta$ is small, the difference between $\pi_{B}$ and $\pi_{C}$
is minimal. Thus, accepting a higher price in $B$-market is perceived
as more costly than improvement of the average quality. Therefore,
a marginal increase in $\theta$ wipes out more traders than it attracts,
leading to a larger decline in the trading volume in the $B$-market
than the increase in $P_{B}$. The resulting $Q$ is, therefore, downward
sloping.

In contrast, when $\theta$ is high, the quality spread, $\Delta\pi$,
becomes significant. Although a higher $\theta$ induces a higher
price $P_{B}$, this does not trigger a large migration since buyers
try to avoid the significant uncertainty in the $C$-market. In this
case, the increment in the price dominates the decline in the transaction
volume in the $B$-market, making the transaction value, $P_{B}K_{B}$,
and the cryptocurrency price, $Q$, increasing in $\theta$.
\begin{figure}[H]
\begin{center}\caption{$\phi = 0.7$}\label{Fig_bigphi}

\includegraphics[scale=0.35]{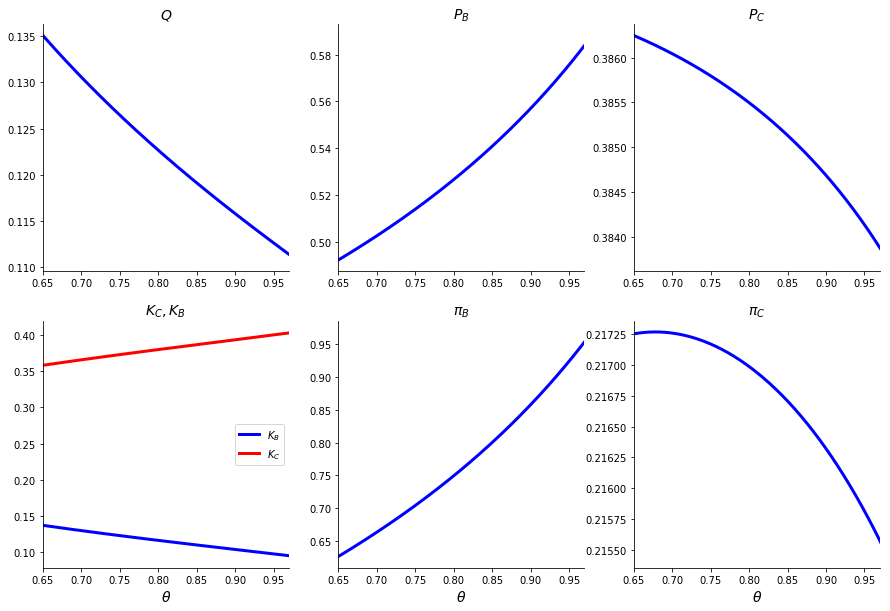}

\end{center}
\end{figure}

\begin{figure}[H]
\begin{center}\caption{$\phi$ = 0.5}\label{Fig_middlephi}

\includegraphics[scale=0.35]{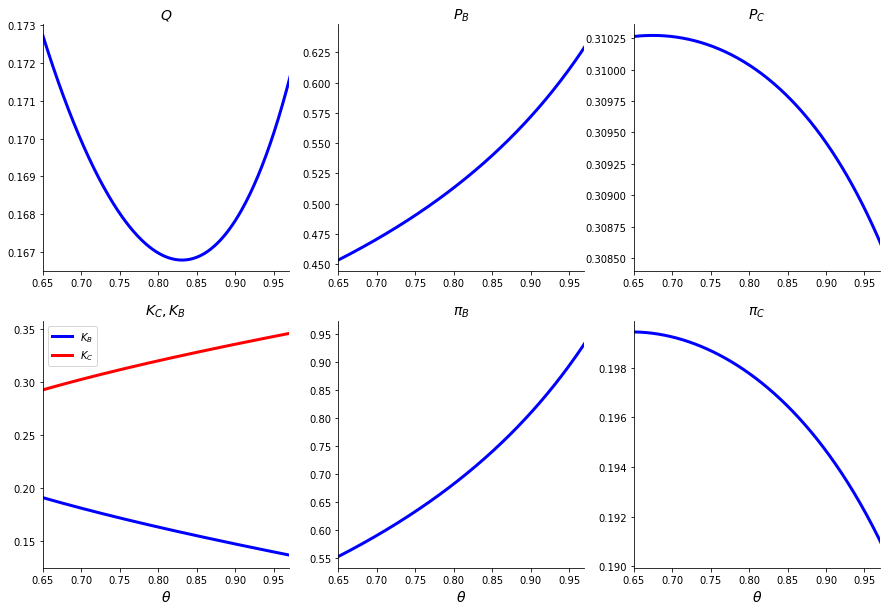}

\end{center}
\end{figure}

\section{Online Appendix : Proof\label{app_proof} }

\subsection{Proof for Proposition \ref{prop_pi} and \ref{prop_price} \label{app_dp_dpi}}

The following argument proves the claim under the generalized model
with $\lambda\in[0,1]$ whose equilibrium conditions are provided
in Appendix \ref{sec:Generalized-Model-with}. Making $\lambda=1$
proves the proposition for the benchmark model.

Our arguments start from two conditions. In the buyers' problem, our
guesses are
\begin{equation}
P_{B}\tilde{\pi}_{C}>P_{C}\tilde{\pi}_{B}\label{g1}
\end{equation}
and
\begin{equation}
\pi_{B}>\pi_{C}.\label{g2}
\end{equation}
Given these, the buyers' partial equilibrium implies that 
\[
\frac{P_{B}}{\tilde{\pi}_{B}}-\frac{P_{C}}{\tilde{\pi}_{C}}=(1-K)+(\tilde{\pi}_{B}-\tilde{\pi}_{C})K_{C}-(1-K)>0,
\]
where $K=K_{B}+K_{C}$. Therefore, we have shown that the inequality
(\ref{g1}) holds in the equilibrium as long as the guess (\ref{g2})
is correct (note that (\ref{g2}) and $\tilde{\pi}_{B}>\tilde{\pi}_{C}$
are equivalent). 

As the next step, we obtain $\pi_{B}$ and $\pi_{C}$ in the general
equilibrium under the guess (\ref{g2}) (and (\ref{g1})). By letting
$\Delta\pi\equiv\pi_{B}-\pi_{C}$ and $F$ be uniform, we have
\begin{equation}
\Delta\pi=\frac{\pi}{K_{B}K_{C}}\left[L-\lambda(1-\lambda)(1-\pi)(1-\theta)\beta_{0}^{U}\frac{\Delta P}{\phi\theta}\right]\label{dpi}
\end{equation}
where $\Delta P=P_{B}-P_{C}$ and 
\[
L=\lambda(1-\pi)\alpha_{I}(P_{B}+\beta_{1}^{U})+(1-\lambda)\beta_{0}^{U}(\lambda(1-\pi)P_{B}+(1-\lambda)(1-\pi_{0})\beta_{1}^{U}),
\]
\[
\beta_{0}^{U}=\frac{P_{C}}{\tilde{\pi}}+\chi\left(\alpha_{0}^{U}-\frac{P_{C}}{\tilde{\pi}}\right),\beta_{1}^{U}=\chi(\alpha_{1}^{U}-\alpha_{0}^{U}).
\]
Since both of $\alpha_{0}^{U}-P_{C}/\tilde{\pi}$ and $\alpha_{1}^{U}-\alpha_{0}^{U}$
are (positively) proportional to $\xi P_{B}-P_{C}$, we have $\beta_{0}^{U}>0$
and $\beta_{1}^{U}\ge0$. Therefore, $L>0$. Moreover, from (\ref{pc}),
the difference in prices is
\begin{align}
\Delta P & =(\tilde{\pi}_{B}-\tilde{\pi}_{C})(1-K_{B})=(1-K_{B})(1-\phi)\Delta\pi,\label{dp_dpi}
\end{align}
where we obviously have $K_{B}<1$. By plugging this into (\ref{dpi}),
we obtain
\[
\Delta\pi=\frac{\pi}{K_{B}K_{C}}\left[L-\lambda(1-\lambda)(1-\pi)(1-\theta)\beta_{0}^{U}\frac{(1-K_{B})(1-\phi)}{\phi\theta}\Delta\pi\right]
\]
\[
\therefore\Delta\pi=\frac{\frac{\pi}{K_{B}K_{C}}L}{1+\frac{\pi}{K_{B}K_{C}}\lambda(1-\lambda)(1-\pi)(1-\theta)\beta_{0}^{U}\frac{(1-K_{B})(1-\phi)}{\phi\theta}}>0.
\]
 Thus, the guess (\ref{g2}) holds in the general equilibrium, and
(\ref{dp_dpi}) implies $P_{B}>P_{C}$.

\subsection{Proof for Proposition \ref{prop_switch}\label{app_switch}\label{subsec:Proof-for-Proposition}}

Suppose that we have $\alpha_{I}>0$. Then the equilibrium solves
\begin{align}
K_{C}^{S} & =(1-\pi)\frac{P_{C}-(1-\theta)P_{B}}{\theta\phi},K_{B}^{S}=\pi P_{B}+(1-\pi)(1-\theta)\left(\frac{P_{B}-P_{C}}{\phi\theta}\right),\label{K_lam1-1}\\
K_{B}^{D} & =1-\frac{P_{B}-P_{C}}{(1-\phi)\pi_{B}},K_{C}^{D}=\frac{P_{B}-P_{C}}{(1-\phi)\pi_{B}}-\frac{P_{C}}{\phi},\nonumber \\
\pi_{B} & =\frac{\pi P_{B}}{K_{B}}.\nonumber 
\end{align}
Let $S=(P_{B}-P_{C})/P_{B}$ be the normalized spread across markets.
Then, rearranging the trading volumes gives
\begin{align*}
K_{B}^{D} & =1-\frac{S}{\pi(1-\phi)}K_{B}^{S},\frac{K_{B}^{S}}{P_{B}}=\pi+(1-\pi)(1-\theta)\frac{S}{\phi\theta},\\
K_{C}^{S} & =\frac{1-\pi}{\phi}P_{B}\left(1-\frac{S}{\theta}\right),K_{C}^{D}=\frac{SK_{B}^{S}}{\pi(1-\phi)}+\frac{P_{B}S}{\phi}-\frac{P_{B}}{\phi}.
\end{align*}
By equating $K_{C}^{S}=K_{C}^{D}$ and substituting $K_{B}^{i}$s,
we get a quadratic equation for $S$. Namely, in the equilibrium,
$S$ solves
\[
\frac{S}{1-\phi}+\frac{(1-\pi)(1-\theta)}{\phi\pi\theta(1-\phi)}S^{2}+\frac{S-1}{\phi}-\frac{1-\pi}{\phi}+\frac{1-\pi}{\theta\phi}S=0.
\]
Note that the LHS is monotonically increasing in $S(\ge0$), and the
condition $\alpha_{I}>0$ is identical to $S<\theta$ by definition
(\ref{eq:ThresSeller}). Thus, in the equilibrium, $\alpha_{I}>0$
if and only if
\[
\frac{\theta}{1-\phi}+\frac{(1-\pi)(1-\theta)}{\phi\pi\theta(1-\phi)}\theta^{2}+\frac{\theta-1}{\phi}-\frac{1-\pi}{\phi}+\frac{1-\pi}{\theta\phi}\theta>0,
\]
which can be rewritten as
\[
\theta^{2}(1-\pi)-\theta+\pi(1-\phi)<0.
\]
Note that if $\theta=0$ then the LHS of this inequality is positive,
while if $\theta=1$ then it is negative. Thus, the smaller solution
of the equation $\theta^{2}(1-\pi)-\theta+\pi(1-\phi)=0$ is between
0 and 1. We set this solution as $\theta_{0}$. Thus, $\alpha_{I}>0$
if and only if $\theta_{0}<\theta\le1.$ 

Next, suppose that $P_{C}-(1-\theta)P_{B}\le0$. This induces $\alpha_{I}=0$
by definition (\ref{eq:ThresSeller}), and the equilibrium solves
\begin{align}
K_{C}^{S} & =0,K_{B}^{S}=\pi P_{B}+(1-\pi)(1-\theta)\frac{P_{B}}{\phi},\label{K_lam1-2}\\
K_{B}^{D} & =1-\frac{P_{B}-P_{C}}{\tilde{\pi}_{B}-\phi},K_{C}^{D}=\frac{P_{B}-P_{C}}{\tilde{\pi}_{B}-\phi}-\frac{P_{C}}{\phi},\nonumber \\
\pi_{B} & =\frac{\pi P_{B}}{K_{B}}.\nonumber 
\end{align}
By using the market clearing in $C$-market and the definition of
$\pi_{B},$ we obtain
\[
K_{B}^{D}=1-\frac{mP_{B}-\phi}{(1-\phi)\pi P_{B}}K_{B}^{S},K_{B}^{S}=\left(\pi+\frac{(1-\theta)(1-\pi)}{\phi}\right)P_{B},
\]
with $m=1+\pi\phi+(1-\pi)(1-\theta)$. By clearing $B$-market, we
have
\begin{align}
P_{B} & =\frac{\phi(\pi+(1-\pi)(1-\theta))}{(2-\theta(1-\pi))(\phi\pi+(1-\pi)(1-\theta))},\label{pb_bd}\\
K_{B} & =\frac{\pi+(1-\pi)(1-\theta)}{2-\theta(1-\pi)}.\label{kb_bd}
\end{align}
Moreover, we can express the market clearing in $B$-market by using
$S$:
\[
K_{B}\left(1+\frac{S}{\pi(1-\phi)}\right)=1.
\]
By plugging the explicit solution of $K_{B},$ we have 
\[
S=\frac{\pi(1-\phi)}{\pi+(1-\pi)(1-\theta)}.
\]
Since $S$ is monotonically increasing in $\theta,$ the condition
$P_{C}-(1-\theta)P_{B}\le0$ is identical to $\theta<S,$ that is
\[
\theta^{2}(1-\pi)-\theta+\pi(1-\phi)\ge0.
\]
Therefore, the condition is $\theta\le\theta_{0}$, and we have established
that the equilibrium is continuous at $\theta=\theta_{0}$.
\begin{cor}
\label{cor_sub}When $\theta\le\theta_{0}$, $P_{B}$, $\pi_{B}$,
$Q$, and $v_{B}$ are monotonically increasing in $\theta$.
\end{cor}
\begin{proof}
Results for $P_{B}$ and $\pi_{B}$ are obvious from (\ref{pb_bd})
and (\ref{K_lam1-2}) in Appendix \ref{app_switch}. By using (\ref{pb_bd})
and (\ref{kb_bd}), we have
\[
Q=\left(\frac{\pi+s}{1+\pi+s}\right)^{2}\frac{\phi}{\phi\pi+s},\ s=(1-\pi)(1-\theta).
\]
Then
\[
\frac{dQ}{ds}\propto2(\phi\pi+s)-(\pi+s)(1+\pi+s)\equiv D_{Q},
\]
and 
\[
(1-\pi)D_{Q}=-\theta^{2}(1-\pi)+\theta-\frac{1+2\pi(1-\phi)}{1-\pi}<0
\]
where the last inequality comes from $\theta\le\theta_{0}.$ With
the fact that $ds/d\theta<0$, we have $dQ/d\theta>0$.
\end{proof}

\subsection{Proof for Proposition \ref{prop_lam1_theta} and \ref{prop_lam1_Q}
\label{app_proof_lam1}}

To see the uniqueness, we plot these $K_{B}$'s against $P_{B}$ (see
Figure \ref{Fig_thetaPB}). Obviously, $K_{B}^{S}$ is positive linear
function in $P_{B}$. We can also check that $K_{B}^{D}$ is concave,
has only one inflection point in $P_{B}>0$, and $\frac{dK_{B}^{D}}{dP_{B}}<0$
for a sufficiently large $P_{B}$. Since $K_{B}^{D}=1>K_{B}^{S}$
at $P_{B}$ such that $K_{B}^{S}=0$, these two curves cross only
once in $P_{B}>0$. 

First, by $K_{B}^{D}+K_{C}^{D}=1-P_{C}/\phi,$ and equating $K_{j}^{D}=K_{j}^{S}$,
we obtain
\begin{equation}
P_{C}=\frac{\phi}{2-\pi}(1-\pi P_{B}).\label{pcpb}
\end{equation}
Now, suppose that $\theta$ increases. This is represented by the
red curves in Figure \ref{Fig_thetaPB}. 
\begin{figure}[H]
\begin{center}\caption{Effect of $\theta$ on $B$-market}\label{Fig_thetaPB}

\includegraphics[scale=0.35]{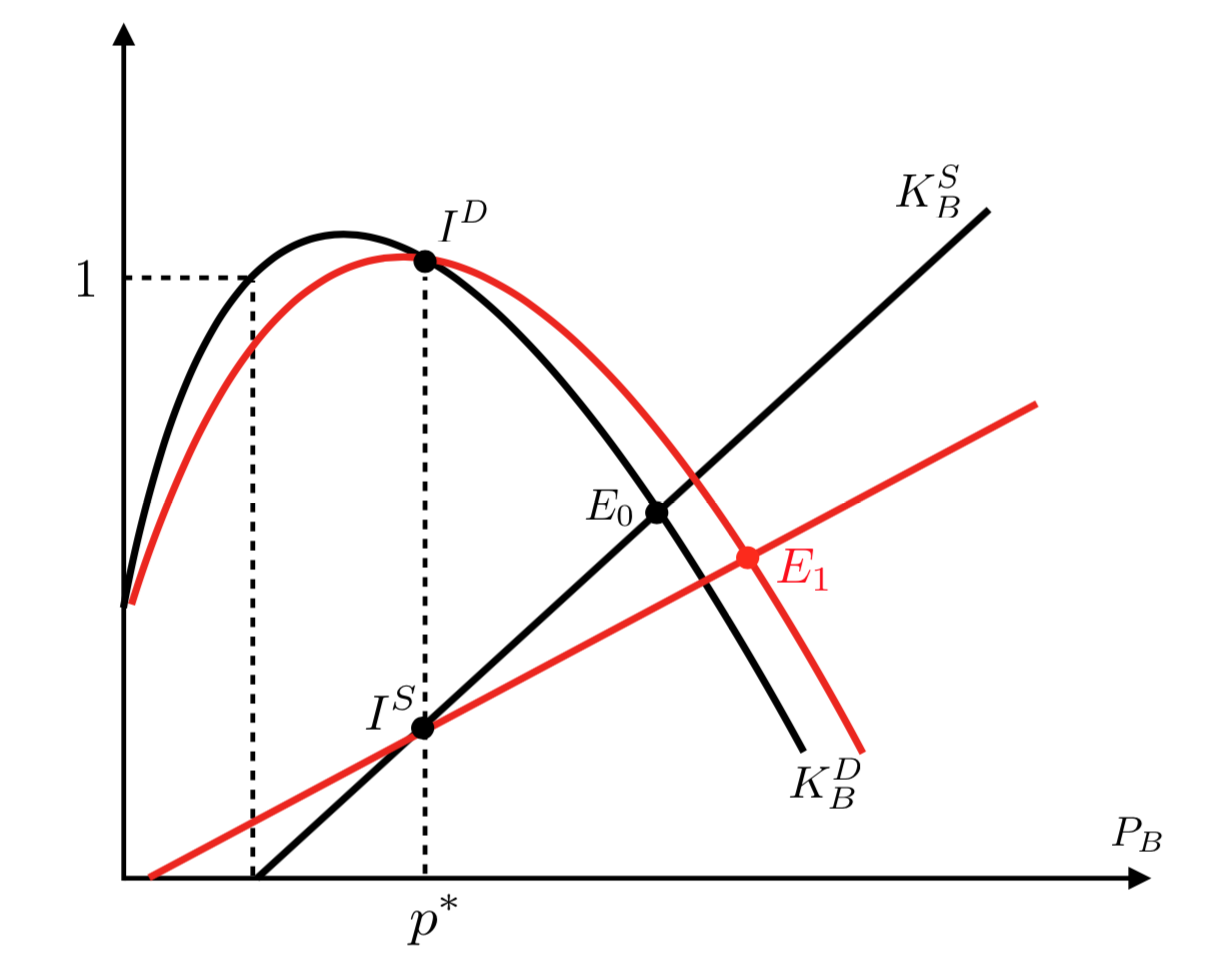}

\end{center}
\end{figure}
 We have 
\begin{equation}
p^{*}=\frac{\phi}{2-\pi(1-\phi)}\label{pstar}
\end{equation}
such that $P_{B}\gtrless P_{C}\Leftrightarrow P_{B}\gtrless p^{*}$.
Also, let $g\equiv\frac{1-\theta}{\theta},\ \eta\equiv1+\frac{\phi\pi}{2-\pi}.$
In the equilibrium, we have $K_{B}^{S}=K_{B}^{D}\equiv K_{B},$ so
that they are respectively expressed as
\begin{align}
K_{B} & =P_{B}\left[\pi+\frac{(1-\pi)}{\phi}g\eta\right]-\frac{1-\pi}{2-\pi}g,\nonumber \\
K_{B} & =\frac{(1-\phi)\pi P_{B}}{((1-\phi)\pi+\eta)P_{B}-\frac{\phi}{2-\pi}}.\label{KB2}
\end{align}
By equating these two equations and rearranging it in terms of $y\equiv P_{B}^{-1}$,
we obtain
\[
H(y,g)\equiv\left(\pi+\frac{1-\pi}{\phi}g\eta\right)-\frac{\pi(1-\phi)y}{((1-\phi)\pi+\eta)-\frac{\phi}{2-\pi}y}-\frac{1-\pi}{2-\pi}gy=0.
\]
For this function, we have
\begin{align}
\frac{\partial H}{\partial g} & =\frac{1-\pi}{\phi(2-\pi)}(2-\pi(1-\phi)-\phi y)>0,\label{H_g}\\
\frac{\partial H}{\partial y} & =-\frac{\pi(1-\phi)((1-\phi)\pi+\eta)}{[((1-\phi)\pi+\eta)-\frac{\phi}{2-\pi}y]^{2}}-\frac{1-\pi}{2-\pi}g<0.\label{H_y}
\end{align}
Note that both inequality comes from $P_{B}>P_{C}$ (equivalently,
$P_{B}>p^{*}$). These confirm, by the implicit function theorem,
$dP_{B}/d\theta>0$. 

We rearrange the equation for $\pi_{B}$ as
\[
\pi_{B}=\frac{\pi}{\pi+\frac{1-\pi}{\phi}g(1-\frac{P_{C}}{P_{B}})},
\]
which implies
\[
\text{sgn}\left(\frac{d\pi_{B}}{d\theta}\right)=-\text{sgn}\left(\frac{d\pi_{B}}{dg}\right)=\text{sgn}\left(\frac{d}{dg}\left[g(1-\frac{P_{C}}{P_{B}})\right]\right).
\]
By using (\ref{pcpb}), we can rewrite the inside of the last brackets:
\begin{align*}
1-\frac{P_{C}}{P_{B}} & =1-\frac{\frac{\phi}{2-\pi}(1-\pi P_{B})}{P_{B}}\propto\frac{2-\pi(1-\phi)-\phi y}{2-\pi}.
\end{align*}
Hence, the last term can be calculated as follows.
\begin{align*}
\frac{d}{dg}\left[g\left(2-\pi(1-\phi)-\phi y\right)\right] & =2-\pi(1-\phi)-\phi y-g\phi\frac{\partial H/\partial g}{\partial H/\partial y}\\
 & =\frac{2-\pi(1-\phi)-\phi y}{\partial H/\partial y}\frac{\pi(1-\phi)((1-\phi)\pi+\eta)}{[((1-\phi)\pi+\eta)-\frac{\phi}{2-\pi}y]^{2}}>0
\end{align*}
where the second line comes from the implicit function theorem, and
the third to last lines are due to (\ref{H_g}), (\ref{H_y}) and
$P_{B}>p^{*}$. Thus, we established that $\frac{d\pi_{B}}{d\theta}>0$.
Also, $K_{B}$ is decreasing in $\theta$, which is immediate from
(\ref{KB2}). The statement (iii) can be checked by the decreasing
$K_{B}$ and $\Delta P/\Delta\tilde{\pi}=1-K_{B}$ in the equilibrium.

As for the price $Q$, (\ref{KB2}) yields
\[
QB_{S}=P_{B}K_{B}=\frac{(1-\phi)\pi P_{B}^{2}}{((1-\phi)\pi+\eta)P_{B}-\frac{\phi}{2-\pi}}.
\]
Since the right hand side does not contain $\theta$, taking a derivative
of the last term is
\[
\frac{dQ}{d\theta}=\frac{dP_{B}}{d\theta}\frac{dQ}{dP_{B}}\propto(\eta+(1-\phi)\pi)P_{B}-\frac{2\phi}{2-\pi}.
\]
Therefore, there is an inflection point
\[
p^{**}=\frac{2\phi}{(\eta+(1-\phi)\pi)(2-\pi)},
\]
which determines the sign of the effect:
\begin{equation}
\frac{dQ}{d\theta}\gtrless0\Leftrightarrow P_{B}\gtrless p^{**}.\label{dQdt}
\end{equation}
Now, by using the implicit formula $H(P_{B}^{-1},g)=0$ and the fact
that $P_{B}H(P_{B}^{-1},g)$ is monotonically increasing in $P_{B}$,
the condition (\ref{dQdt}) is identical to 
\[
A(\theta)\equiv g(1-\pi)\left(2\eta-h\right)+2\pi[\phi-(1-\phi)(2-\pi)]\lessgtr0.
\]
Note that $A$ is monotonically decreasing in $\theta$ (this can
be confirmed by using $P_{B}>p^{*}$ again). By letting $\theta$
fluctuate from $\theta_{0}$ to $1$, we have the following result.

(i) If $\phi>(2-\pi)/(3-\pi),$ then $A(\theta)>0$ for all $\theta\in[\theta_{0},1]$,
which implies that $P_{B}>p^{**}$ always holds in the equilibrium,
leading to a monotonically decreasing $Q$. (ii) If $\phi\le(2-\pi)/(3-\pi),$
then $A(1)<0$, so $Q$ is decreasing in high-$\theta$ region. To
understand more global behavior, we need to check if $A(\theta_{0})\gtrless0$.
By seeing $A$ as a function of $g$, we can define $g^{*}$ that
makes $A(g)=0$ as
\[
g^{*}(\phi)=\frac{2\pi(2-\pi-\phi(3-\pi))}{(1-\pi)(1-\pi(1-\phi)+\frac{\phi\pi}{2-\pi})}.
\]
Since $A(g)$ is increasing in $\phi$, we have $dg^{*}/d\phi<0$.
Note that we are focusing on $\theta>\theta_{0}$, which means 
\[
g<g_{0}(\phi)\equiv\frac{1-\theta_{0}(\phi)}{\theta_{0}(\phi)}.
\]
From the definition of $\theta_{0},$ we know $\theta_{0}$ is decreasing
and $g_{0}$ is increasing in $\phi$. We also have $\lim_{\phi\rightarrow0}g^{*}(\phi)>0$
and $\lim_{\phi\rightarrow0}g_{0}(\phi)=\mathbb{I}_{\{\pi<1/2\}}\pi^{-1}$
because $\theta_{0}\rightarrow\mathbb{I}_{\{\pi\ge1/2\}}+\mathbb{I}_{\{\pi<1/2\}}\frac{\pi}{1-\pi}$.
Figure \ref{Fig_Q_proof} shows the effect of a smaller $\phi$ on
$g^{*}$ and $g_{0}$. We have following two possibilities. 
\begin{figure}[H]
\begin{center}\caption{Behavior of $A$}\label{Fig_Q_proof}

\includegraphics[scale=0.3]{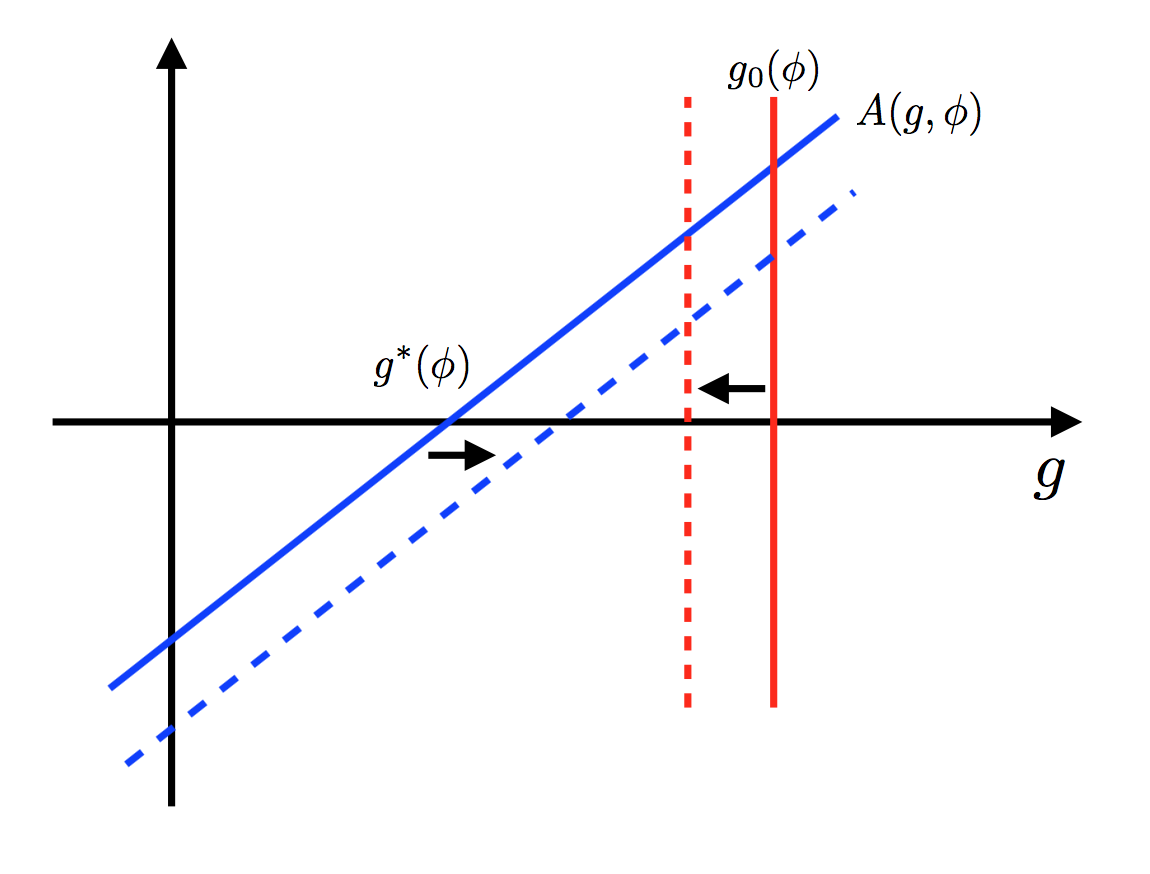}

\end{center}
\end{figure}

\begin{description}
\item [{(ii-a)\label{ii-a}}] Suppose that $\pi\ge1/2$. Then there is
$\phi_{0}$ that solves 
\begin{equation}
g^{*}(\phi)=g_{0}(\phi).\label{app_phi1}
\end{equation}
$\phi_{0}$ is uniquely determined from the discussion above. In this
case, if $\phi<\phi_{0}$, then $A(g)<0$ for all $g<g_{0}$. That
is $Q$ is monotonically increasing in $\theta$. If $\phi_{0}<\phi<\phi_{1}$,
then we have $A(g)\gtrless0\Leftrightarrow g\gtrless g^{*}.$ Thus,
we can define $\theta^{*}=1/(1+g^{*})\in(\theta_{0},1]$, and $Q$
is increasing when $\theta>\theta^{*}$ and decreasing when $\theta<\theta^{*}$. 
\item [{(ii-b)}] Also, consider the case with $\pi<1/2$. In this case,
we have a unique $\pi^{*}\in(0,1/2)$ that solves
\[
g^{*}(0)=\frac{2\pi(2-\pi)}{(1-\pi)^{2}}=\frac{1}{\pi}=g_{0}(0),
\]
or equivalently 
\[
2\pi^{3}-3\pi^{2}-2\pi+1=0.
\]
If $\pi^{*}\le\pi<1/2,$ then $g^{*}(0)>g_{0}(0)$. This implies that
we always have $\theta^{*}$ defined by $\theta^{*}=1/(1+g^{*})\in(\theta_{0},1]$,
and and $Q$ is increasing when $\theta>\theta^{*}$ and decreasing
when $\theta<\theta^{*}$. On the other hand, if $0\le\pi\le\pi^{*}$,
then the arguments go back to the case \textbf{(ii-a)} and the same
results hold.
\end{description}

\section{Welfare Analyses}

\subsection{Welfare Analyses for Buyers under $\lambda=1$\label{app_welfare_b}}

The buyers' welfare in aggregate is 
\begin{align}
v_{B} & =\int_{\alpha^{*}}(\Delta\tilde{\pi}\alpha-\Delta P)dF+\int_{P_{C}/\phi}(\phi\alpha-P_{C})dF\nonumber \\
 & =\frac{\Delta\tilde{\pi}}{2}(1-\alpha^{*})^{2}+\frac{\phi}{2}(1-\frac{P_{C}}{\phi})^{2}\nonumber \\
 & =\frac{1}{2}\left[\pi(1-\phi)P_{B}K_{B}+\frac{\phi}{(2-\pi)^{2}}(1-\pi+\pi P_{B})^{2}\right],\label{vb_Q}
\end{align}
where the second term comes from $\alpha^{*}=\Delta P/\Delta\tilde{\pi}$,
and the last term comes from (\ref{pcpb}), (\ref{K_lam1}), $K_{j}^{S}=K_{j}^{D}$,
and the definition of $\pi_{B}$:
\[
\Delta\tilde{\pi}=(1-\phi)\frac{\pi P_{B}}{K_{B}}.
\]
The property of the first term is given by Proposition \ref{prop_lam1_Q},
while the second term is monotonically increasing in $\theta$. Furthermore,
by using (\ref{KB2}),
\[
2v_{B}=((1-\phi)\pi)^{2}\frac{P_{B}^{2}}{P_{B}[\eta+\pi(1-\phi)]-\frac{\phi}{2-\pi}}+\frac{\phi}{(2-\pi)^{2}}(1-\pi+\pi P_{B})^{2}.
\]
Note that $\theta$ does not directly affect $v_{B}$ in this expression.
By letting $h\equiv\eta+\pi(1-\phi),$ we have
\[
\frac{d2v_{B}}{dP_{B}}\equiv D_{B}=((1-\phi)\pi)^{2}\frac{P_{B}(hP_{B}-2\frac{\phi}{2-\pi})}{(hP_{B}-\frac{\phi}{2-\pi})^{2}}+\frac{2\phi\pi}{(2-\pi)^{2}}(1-\pi+\pi P_{B}).
\]
The second order derivative yields
\begin{align*}
\frac{dD_{B}}{dP_{B}} & =\frac{2((1-\phi)\pi)^{2}}{(hP_{B}-\frac{\phi}{2-\pi})^{3}}\left[\left(hP_{B}-\frac{\phi}{2-\pi}\right)^{2}-hP_{B}\left(hP_{B}-\frac{2\phi}{2-\pi}\right)\right]+\frac{2\phi\pi^{2}}{(2-\pi)^{2}}\\
 & =\frac{2((1-\phi)\pi)^{2}}{(hP_{B}-\frac{\phi}{2-\pi})^{3}}\frac{\phi^{2}}{(2-\pi)^{2}}+\frac{2\phi\pi^{2}}{(2-\pi)^{2}}>0.
\end{align*}
We also have $P_{B}(\theta=1)\equiv\hat{p}_{1}=\frac{1-\phi+\frac{\phi}{2-\pi}}{h}$
and can check $D_{B}(P_{B}=\hat{p}_{1})>0$. Thus, if $\lim_{\theta\rightarrow\theta_{0}}D_{B}<0$,
there is a unique $\theta^{*}$ such that $D_{B}\gtrless0\Leftrightarrow\theta\gtrless\theta^{*}$,
while if $\lim_{\theta\rightarrow\theta_{0}}D_{B}>0$, then $D_{B}>0$
for all $\theta$. 

The following formulas at $\theta=\theta_{0}$ simplify the analyses.
First, as $\theta\searrow\theta_{0},$ we have
\begin{equation}
P_{C}=\frac{\phi}{2-\pi}\left(1-\pi P_{B}\right)=(1-\theta_{0})P_{B},\label{t_big_t0}
\end{equation}
\begin{equation}
\therefore P_{B}=\tilde{p}\equiv\frac{\phi}{\pi\phi+(1-\theta_{0})(2-\pi)}.\label{tildep}
\end{equation}
Moreover, at $\theta\rightarrow\theta_{0},$ we have $\alpha_{I}\rightarrow0$
by definition. Since the markets have to clear, at the limit,
\begin{align}
\lim_{\theta\searrow\theta_{0}}K_{B}^{D} & =\lim_{\theta\searrow\theta_{0}}\left(1-\frac{P_{B}-P_{C}}{\pi_{B}(1-\phi)}\right)=\lim_{\theta\searrow\theta_{0}}\left(1-K_{C}^{D}-\frac{P_{C}}{\phi}\right)\nonumber \\
 & =\lim_{\theta\searrow\theta_{0}}\left(1-(1-\pi)\alpha_{I}-\frac{P_{C}}{\phi}\right)\nonumber \\
 & =1-\frac{1-\theta_{0}}{\phi}\tilde{p}=\frac{1-\pi+\pi\tilde{p}}{2-\pi}\label{K0}
\end{align}
The first line is the definition, the second line is from the definition
of $K_{C}^{D},$ the third line is from the market clearing condition
in $C$-market, and the fourth and fifth lines are from the definition
of $\theta_{0}$ that gives $\alpha_{I}=0$ and (\ref{t_big_t0}).
The last line is the other expression from (\ref{t_big_t0}). Also,
from (\ref{KB2})
\begin{equation}
\lim_{\theta\searrow\theta_{0}}K_{B}=\frac{(1-\phi)\pi\tilde{p}}{((1-\phi)\pi+\eta)\tilde{p}-\frac{\phi}{2-\pi}}.\label{K1}
\end{equation}
Since markets have to clear, all of these expressions (\ref{K0},
\ref{K1}) have to be identical. That is
\begin{equation}
\frac{(1-\phi)\pi\tilde{p}}{h\tilde{p}-\frac{\phi}{2-\pi}}=1-\frac{1-\theta_{0}}{\phi}\tilde{p}=\frac{1-\pi+\pi\tilde{p}}{2-\pi},\label{three}
\end{equation}
at (\ref{tildep}). 

Let $D_{B,0}\equiv\lim_{\theta\searrow\theta_{0}}D_{B}$. By using
the equality of the first and last term in (\ref{three}),
\begin{align*}
D_{B,0} & \propto(1-\phi)\pi\left(1-\frac{\frac{\phi}{2-\pi}}{h\tilde{p}-\frac{\phi}{2-\pi}}\right)+\frac{2\phi\pi}{2-\pi}.
\end{align*}
By using (\ref{three}) once again, $h\tilde{p}-\frac{\phi}{2-\pi}=\frac{(1-\phi)\pi\tilde{p}}{1-\frac{1-\theta_{0}}{\phi}\tilde{p}}$.
Thus,
\begin{align*}
D_{B,0} & \propto(1-\phi)\pi\left(1-\frac{\phi}{2-\pi}\frac{1-\frac{1-\theta_{0}}{\phi}\tilde{p}}{(1-\phi)\pi\tilde{p}}\right)+\frac{2\phi\pi}{2-\pi}\\
 & \propto[(1-\theta_{0})+(2-\pi)(1-\phi)]+2\pi\phi-\frac{\phi}{\tilde{p}}\\
 & =1+(1-\pi)(\theta_{0}-2\phi).
\end{align*}
Note that $\theta_{0}$ is decreasing function of $\phi$ and $\lim_{\phi\rightarrow1}\theta_{0}=0$.
Then, $\min_{\phi}D_{B,0}=\lim_{\phi\rightarrow1}D_{B,0}=2\pi-1.$
Therefore, if $\pi<\frac{1}{2}$, we can define a unique $\phi=\phi_{2}$
that solves 
\[
1+(1-\pi)(\theta_{0}-2\phi)=0.
\]
If $\phi\le\phi_{2}$ or $\pi>1/2$, then $D_{B}>0$ for all $\theta\in(\theta_{0},1]$,
and $v_{B}$ is monotonically increasing. On the other hand, if $\pi\le1/2$
and $\phi>\phi_{2}$, then there is a unique $\theta^{**}$ such that
$D_{B}\gtrless0\Leftrightarrow\theta\gtrless\theta^{**}$.

\subsection{Welfare of Sellers}

The welfare of sellers hinges on the quality of assets they are allocated
upon their arrival at the economy. The aggregate welfare of $H$-
and $L$-type sellers are defined as 
\begin{align*}
v_{S,H} & =\int_{P_{B}}\alpha dF+\int^{P_{B}}P_{B}dF,\\
v_{S,L} & =\int_{\frac{P_{B}}{\phi}}\phi\alpha dF+\int_{\alpha^{I}}^{\frac{P_{B}}{\phi}}((1-\theta)P_{B}+\theta\phi\alpha)dF+\int^{\alpha^{I}}P_{C}dF.
\end{align*}
In both expressions, the first term is the welfare of inactive sellers,
while the second term is the welfare of sellers in $B$-market. The
last term of $v_{S,L}$ comes from sellers of $L$-asset in $C$-market.
It is easy to check the following proposition:
\begin{prop}
\label{prop:The-welfare-of-1-1}The welfare of sellers with high-quality
assets, $v_{S,H}$, is monotonically increasing in $\theta$. 
\end{prop}
The innovation in blockchain always benefits sellers of high-quality
assets because they can always sell their asset in $B$-market at
the higher price $P_{B}$. On the other hand, the global effect of
$\theta$ on $v_{S,L}$ is hard to determine, though we can obtain
the following local result:
\begin{prop}
\label{prop:-cannot-be-1-1}$\left.\frac{dv_{S,L}}{d\theta}\right|_{\theta=1}<0$,
that is, $\theta=1$ cannot be the maximizer of $v_{S,L}$.
\end{prop}
\begin{proof}
See Appendix \ref{subsec:Welfare-Gain-for Sellers}.
\end{proof}
Together with Proposition \ref{prop:The-welfare-of-1-1}, this implies
that ex-post welfare of $L$-type sellers cannot agree with the welfare
of $H$-type sellers regarding the optimal $\theta$. 

To obtain more intuitions, we can separate $v_{S,L}$ into the welfare
gain parts and the reservation welfare, as in the case of buyers'
welfare, 
\begin{equation}
v_{S,L}=P_{C}+\int_{\alpha^{I}}^{1}((1-\theta)P_{B}-P_{C}+\theta\phi\alpha)dF+\int_{\frac{P_{B}}{\phi}}^{1}(1-\theta)(\phi\alpha-P_{B})dF.\label{vsl-1-1}
\end{equation}
First, all the $L$-type sellers certainly can obtain the reservation
welfare of $P_{C}$ by selling the asset in $C$-market (the first
term in \ref{vsl-1-1}). If $\alpha>\alpha_{I}$, the sellers change
the behavior to either selling in $B$-market or keeping it. The second
term in (\ref{vsl-1-1}) represents the welfare gain of sellers who
will opt-out from $C$-market: all of them ($\alpha>\alpha_{I}$)
can potentially obtain the additional welfare by selling in $B$-market.
Within this subgroup, agents with relatively high $\alpha$ (such
that $\alpha>\frac{P_{B}}{\phi}$) prefer to keep the asset by giving
up the revenue $P_{B}$, which yields the further welfare gain exhibited
by the last term of $v_{S,L}$. 

The first and last terms are monotonically decreasing in $\theta$.
That is, the reservation welfare (the first term) and the gain from
changing behavior from selling in $C$-market to being inactive decline
as the blockchain market becomes more profitable. The sign of the
impact on the middle term is affected by two competing effects. On
one hand, a higher $\theta$ boots the revenue by heightening $P_{B}$.
On the other hand, it reduces the expected revenue by making the rejection
risk higher. The total effect depends on how large the positive welfare
gain by the traders in $B$-market will be, and it is more likely
to happen when the migration of buyers from $B$-market is not so
large due to the severe information asymmetry and large quality spread.
In Appendix \ref{subsec_seller_gain}, we provide further analyses
regarding the welfare gain of sellers to complement Proposition \ref{prop:-cannot-be-1-1}
and show that the effect of $\theta$ on $v_{S,L}$ depends on the
elasticity of $P_{B}$ with respect to $\theta$.

\subsection{Welfare Gain for Sellers\label{subsec_seller_gain}\label{subsec:Welfare-Gain-for Sellers}}

Hypothetically, consider a randomly picked seller who is deprived
of the access to $B$-market. Ex-ante (before she is endowed with
the asset), she expects to have $v_{S}^{0}=\pi v_{S,H}^{0}+(1-\pi)v_{S,L}^{0}$,
where $v_{S,i}^{0}$ represents the reservation welfare of the seller
when she obtains the asset-$i$ with $i\in\{L,H\}$. Specifically,
\begin{align*}
v_{S,i}^{0} & =\begin{cases}
\int_{0}^{P_{C}}P_{C}dF+\int_{P_{C}}^{1}\alpha dF & \text{for \ensuremath{i=H}}\\
\int_{0}^{\frac{P_{C}}{\phi}}P_{C}dF+\int_{\frac{P_{C}}{\phi}}^{1}\phi\alpha dF & \text{for \ensuremath{i=L.}}
\end{cases}
\end{align*}
Note that $P_{C}$ is the equilibrium price in the segmented market
economy rather than the single market economy since we consider a
non atomic agent who does not have any impact on the segmented market
equilibrium.

For this agent, the expected welfare gain from having access to $B$-market
is given by 
\begin{align}
\Delta v_{S} & =v_{S}-v_{S}^{0}\nonumber \\
 & =\pi\Delta v_{S,H}+(1-\pi)\Delta v_{S,L}.\label{eq:Dv_S}
\end{align}
By applying uniform assumption, we obtain simple formulae:
\begin{align}
\Delta v_{S,j} & =\begin{cases}
\frac{1}{2}(P_{B}^{2}-P_{C}^{2}) & \text{if }j=H\\
(1-\theta)\frac{\Delta P^{2}}{2\phi\theta} & \text{if \ensuremath{j=L}.}
\end{cases}\label{eq:dvsh}
\end{align}
Obviously, $\Delta v_{S,H}$ is monotonically increasing in $\theta$
since $P_{C}$ is monotonically decreasing in $\theta$. Intuitively,
the \textit{reservation} welfare for sellers with $H$-asset is decreasing
in $\theta$ since the terms of trade in $C$-market will deteriorate
if the blockchain technology improves. That is, the more secure the
blockchain becomes, the larger the gain from having the access to
$B$-market will be for $H$-type sellers.

On the other hand, as for $\Delta v_{S,L}$, we have 
\[
\frac{d\Delta v_{S,L}}{d\theta}=\frac{\Delta P}{2\phi}\left[\Delta P\frac{dg}{d\theta}+2g\frac{d\Delta P}{d\theta}\right]
\]
where $g=(1-\theta)/\theta$. Since, $P_{C}=\phi(1-\pi P_{B})/(2-\pi)$
in the equilibrium, it becomes 
\begin{equation}
\frac{d\Delta v_{S,L}}{d\theta}=\frac{\Delta P}{2\phi}\frac{dg}{d\theta}\frac{1}{2-\pi}\left[[2-\pi(1-\phi)]P_{B}(1-\varepsilon_{P})-\phi\right],\text{ }\varepsilon_{P}\equiv-\frac{dP_{B}/dg}{P_{B}/g}>0.\label{eq:elasticity}
\end{equation}
$\varepsilon_{P}$ represents the elasticity of $P_{B}$ regarding
the change in the security $\theta$ (since $g$ and $\theta$ have
negative monotone relationship, we consider it as the effect of $\theta$).
When the elasticity is high, \emph{i.e.}, $\varepsilon_{P}$ is large,
$\Delta v_{S,L}$ is increasing in $\theta$. Otherwise, it is decreasing
in $\theta$. Since the welfare gain of sellers with $L$-asset comes
only from the transaction through $B$-market, a higher $\theta$
has two competing effects. First, a higher $\theta$ increases the
offer price $P_{B}$ in $B$-market, which has a positive impact on
the sellers' welfare through a higher return from selling. This higher
$P_{B}$ proliferates the positive impact on $\Delta v_{S,L}$ by
inducing a higher probability of submitting selling order into $B$-market.
On the other hand, higher security level makes the rejection probability
higher. This effect reduces the gain for sellers with $L$-asset.
Given that the latter effect is a direct consequence of $\theta$,
the first positive effect dominates the latter effect when the increment
of $P_{B}$ is large, namely, $\varepsilon_{P}$ is high. 

When is the elasticity more likely to be high? It can be translated
into the market equilibrium: a higher $P_{B}$ confounds the demand
when the cost of migration for the buyers is low. On the other hand,
if the cost of migration is high, the higher price can sustain itself,
making the elasticity of $P_{B}$ high. Thus, $\Delta v_{S,L}$ exhibits
upward sloping curve when (i) $\phi$ is low or (ii) $\Delta\pi$
is large.

\section{Technology Overview\label{sec:Technology-Overview:-Cryptocurre-1}}

\subsection{Bitcoin}

The leading example of cryptocurrency is Bitcoin. The idea of Bitcoin
is first introduced by \citet{nakamoto2008bitcoin}, who proposes
the blockchain technology for the first time. The part of the objectives
of this proposal is to offer a solution to the ``double spending''
problem. Bitcoin is the first success after a long history of proposals
of decentralized media of transactions, making it the largest market
capitalization in the cryptocurrency trading market (\citealp{narayanan2016bitcoin}).\footnote{As of February 8, 2018} 

The Bitcoin blockchain has recorded the information of the flow of
bitcoins across participants (``Alice paid X bitcoin to Bob'') in
a tamper-proof manner. In this platform, the traded good is bitcoin
itself. To have a concrete idea, we aline the transaction of bitcoin
with the example in Appendix \ref{sec:Appendix:-Motivating-Examples}. 

Suppose that, at date $t=0$, liquidity providers have liquid assets
(cash or bitcoin), and liquidity takers are endowed with illiquid
assets whose common value is $k$. At date $t=1$, takers are hit
by a liquidity shock and want to offload (liquidate) their asset holding
to obtain (net) utility $vs-k$ from $s$ amount of liquid cash (or
coin), where $v$ is some positive private value. The state of liquidity
providers at date $t=1$ is either $s_{t}\in\{m,0\}$, where $s_{t}$
represents the amount of cash or bitcoin she holds. 

To bridge the argument to the example introduced in Appendix \ref{sec:Appendix:-Motivating-Examples},
we can think of the realization of $s_{1}$ as the result of cumulative
transactions: there are dates $t\in\{-N,-(N-1),\cdots,-1,0\}$, and
each has a state $s_{t}$ that represents cash flow at each date. 

Suppose that $s_{1}=0$ realizes (she already spent her cash or coin
in the past) for $1-\pi$ fraction of liquidity providers, and rest
of them have $s_{1}=m$.\footnote{Assume that $\pi$ is common knowledge, and $m<k<\pi vm$ hold so
that transactions take place.} Since announcing $\hat{s}_{1}=m$ is strictly dominant for all of
the buyers, $1-\pi$ fraction of them are fraudulent who attempt to
use the coin or money they already spent (double-spend). In the traditional
cash market without a bank that monitors the accounts of her customers
and transactions, fraudulent agents easily spend their money twice
(or more) as long as $\pi vm>k$, because this inequality means that
sellers want to sell the asset. 

On the other hand, in the Bitcoin's network, it is extremely difficult
to spend the coin twice because even if the agent with $s_{1}=0$
claims $\hat{s}_{1}=m$, this cannot be an agreement. That is, $\theta$
fraction of $1-\pi$ agents fail to accomplish their fraud transactions,
which makes the fraction of honest sellers $\pi(\theta)=\frac{\pi}{\pi+(1-\theta)(1-\pi)}>\pi$.
This provides a higher expected return for liquidity takers, and hence
they have an incentive to utilize the blockchain platform rather than
the traditional transaction.\footnote{Of course, knowing that $\theta$ fraction of ``double spending''
fails, the behavior of liquidity providers also changes. We do not
go into detail of this point here and leave it for the formal analyses
in Section 3 of the main text.} In the example of Bitcoin, $\theta$ is very high since the double
spending is precluded unless an agent has a prohibitively strong computing
power. 

\subsection{Ethereum}

As explained in the main text, precluding double spending is not the
only feature enabled by the blockchain technology. It also allows
us to write complex scripts to determine what kind of information
is regarded and added as ``relevant'' one.\footnote{Technically speaking, the language for the scripts in Ethereum transactions
is Turing-complete, a class of language that allows complex statements.
This capacity of allowing complexity makes state-contingent contracts
possible.} First, as mentioned earlier, the state information recorded on the
blockchain is highly credible. Second, by writing codes such that
``transaction takes place if and only if the state $s$ satisfies
conditions (1)..., (2)...., and (N),'' we can make a transaction
contingent on desirable conditions (1)-(N). This technology has many
applications to mitigate informational problem in assets transactions,
information storage, and allocation. 

The wine blockchain, founded by EY Advisory \& Consulting Co. Ltd.
offers a concrete example.\footnote{Other sectors outside the financial service industry, such as supply
chain management, are also interested in digitizing and tracking information
of products. For example, Walmart Stores Inc. is testing a transaction
system on IBM blockchain technology to manage supply-chain data (\url{https://www-03.ibm.com/press/us/en/presskit/50610.wss},
visited on May 10, 2018). The products include porks, mangoes, berries
and a dozens of other products. It is aimed to identify bad sources
throughout the overall chains of food product intermediations. } Traditionally, wine market is exposed to a risk of counterfeit (``lemons''
in the sense of \citeauthor{akerlof1970market} {[}\citeyear{akerlof1970market}{]}),
whose economic losses is said to be \$1-5 billion per year.\footnote{For example, see \citet{holmberg2010wine} and \citet{przyswa2014counterfeiting}.}
The problem of low-quality wines is severe since many intermediaries
are involved in a supply chain of wine, making it difficult to keep
track of all the transactions from ingredient firms to retail stores.
By utilizing the blockchain and smart contract, however, transactions
of wines become almost free from the lemons' problem without any credible
third-parties or interventions.\footnote{See \citet{buterin2016ethereum} for more details.} 

In contrast to the example of Bitcoin, which records a flow of coin
as a state variable, this example can be directly adopted to the preceding
example in Appendix \ref{sec:Appendix:-Motivating-Examples}. That
is, state information can take a range of characteristics: it can
record a brand of ingredient, name of wine-producer, in what temperature
and how long a wine has been stored, and so on. This can also be applied
to other classes of assets whose value is difficult to identify for
consumers. Given the descriptions of state information, Ethereum allows
us to make transactions conditional on realization of desirable states. 

\subsection{Connection of Blockchain and Cryptocurrency}

The two examples the blockchain above use cryptocurrency as a means
of transaction. This class of blockchain platforms includes the one
for transactions of wines (EY, based on Ethereum), security (tZERO),
international remittance (Bitcoin), arts and photography (Kodak, based
on KodakOne and KodakCoin), and more. 

Another interesting example is Ripple. Although their underlying technology
is not exactly the blockchain, Ripple also utilizes a distributed
ledger to provide secure transactions between banks and commercial
firms, in which cryptocurrency XRP is used. An approval of transaction
is not made by Proof of Work (PoW) as in the Bitcoin system, but it
is done by a certified set of validating nodes. Hence a transaction
is settled faster than in the Bitcoin system, and waste of electricity
inherent in the PoW system is relaxed. 

There are also blockchain platforms that do not need circulation of
cryptocurrencies as a medium of transactions. For instance, EverLedger,
providing the blockchain platform for exchanging a variety of assets
(wine, art, jewelry), claims that they have no interest in building
their own cryptocurrency since they want to avoid many political challenges.\footnote{\href{urlto:http://www.eweek.com/cloud/hyperledger-blockchain-project-is-not-about-bitcoin}{http://www.eweek.com/cloud/hyperledger-blockchain-project-is-not-about-bitcoin}}
Moreover, a number of ``permissioned blockchain'' platforms do not
need to use digital currency or mining process to record information. 

For a blockchain platform whose transactions are not necessarily executed
by cryptocurrency, the model provides an implication for the fundamental
value (price) of the blockchain itself. Specifically, Proposition
\ref{prop:The-fundamental-price} of the main text proposes a theoretical
measure for the price of these types of blockchain technology and
show that it corresponds to the welfare gain of participants in the
network.

\section{Imposing Fee on Sellers}

What if the suppliers (or producers of goods) have to pay the fee
to use the blockchain? We can fall back on the same logic to derive
the maximum possible fee that the manager can charge on sell-side
of the market, which we denote as $f_{S}$:
\begin{align*}
f_{S} & =\Delta v_{S}
\end{align*}
with $\Delta v_{S}$ in (16) of the main model. 

When $\Delta v_{S,L}$ is increasing in $\theta$, the total fee $f_{S}$
is also increasing, while if $\Delta v_{S,L}$ is decreasing, the
form of $f_{S}$ is ambiguous since it depends on the level of $\pi$.
Suppose that the parameter values make $\Delta v_{S,L}$ decreasing
in $\theta$. Under this situation, it seems natural to conclude that
a lower $\text{\ensuremath{\pi}}$ makes $f_{S}$ downward sloping
because it puts more weight on $\Delta v_{S,L}$. However, this is
not necessarily the case. For example, if we make $\pi\rightarrow0$,
we have $d\Delta v_{S,L}/d\theta\rightarrow0$ and $df_{S}/d\theta\rightarrow0$.
This is because of the dominating $L$-asset in the market. As the
level of $\pi$ diminishes, the share of $L$-asset increases, and,
at the limit, there are only $L$-asset in both of the markets. This
implies that having the access to $B$-market does not payout: the
welfare gain converges to zero. Further analyses on the sellers' willingness
to pay are provided in Appendix D as numerical experiments because
of the difficulty of an analytical characterization.

\subsection{Manager vs. Sellers\label{subsec:Manager-vs.-Sellers}}

Next, suppose that the manager makes money by imposing the fee on
the sell-side of the market, while the government tries to maximize
sellers' welfare. From (16), (17), and (18), we know the followings:
the welfare gain of $H$-asset holders is monotonically increasing,
imposing a positive pressure on $f_{S}$, while that of $L$-asset
holders has an ambiguous effect. Moreover, Proposition 8 implies that,
as long as $f_{S}$ is monotonically increasing in $\theta$ due to
the dominating effect from $d\Delta v_{S,H}/d\theta$, the optimal
$\theta$ set by the manager cannot agree with $\theta$ that maximizes
$v_{S,L}$. Thus, the welfare of $H$-asset holders is maximized,
while sellers with $L$-quality asset incur welfare loss. Proposition
8 implies that there is a conflict between the welfare of $H$-asset
holders and $L$-asset holders: once the type of asset is realized,
even the social planner cannot maximize the welfare of both types
of traders. 

Since we cannot analytically characterize the properties of sellers'
welfare further, we rely on the numerical examples. We find the total
fee revenue is upward sloping, and the maximizing $f_{S}$ agrees
with maximizing $v_{S,H}$ in most ranges of parameters. Note that
the discussion around elasticity makes clear in what situation this
welfare loss tends to occur. 

\subsection*{Figures and Tables\label{sec:Figures-and-Tables}}

The following figures provide the numerical examples for the sellers'
welfare and fees. Parameters take $\pi\in\{0.01,0.1,0.4,0.7,0.9\},\text{ and }\phi\in\{0.35,0.5,0.7\}$.
The first (second) column shows the total and reservation welfare
of $H$-type ($L$-type) sellers, as well as the fee imposed by the
manager, $f_{S}$. The third column is the plot of the total (ex-ante)
welfare of sellers and $f_{S}$. 

As suggested by the theory, $v_{H}$ is monotonically increasing in
$\theta$, while $v_{S,L}$ is either monotonically decreasing or
hump-shaped. If $f_{S}$ can be decreasing in $\theta$, that should
occur when $\theta$ is relatively high. However, even if we set the
share of $L$-type sellers large ($\pi=0.01$), the configuration
of $f_{S}$ is upward sloping. This is because the change in $\theta$
affects $v_{S,L}$ mostly through the change in $v_{S,L}^{0}$ when
$\pi$ is small. That is, the welfare gain for $L$-type seller, $v_{S,L}-v_{S,L}^{0}$,
is not affected by $\theta$ (see the difference between blue and
green-dotted lines). Of course, a higher $\theta$ increases the welfare
gain from trading in $B$-market. Meanwhile, it reduces the welfare
gain in $C$-market, which has a dominating effect on the welfare
because only $1-\theta$ fraction of selling attempt get to have a
benefit of a higher $\theta$.  

\begin{figure}[H]
\begin{center}\caption{Welfare of Sellers and Fee: $\pi = 0.01$}\label{Fig_sell}

\includegraphics[scale=0.15]{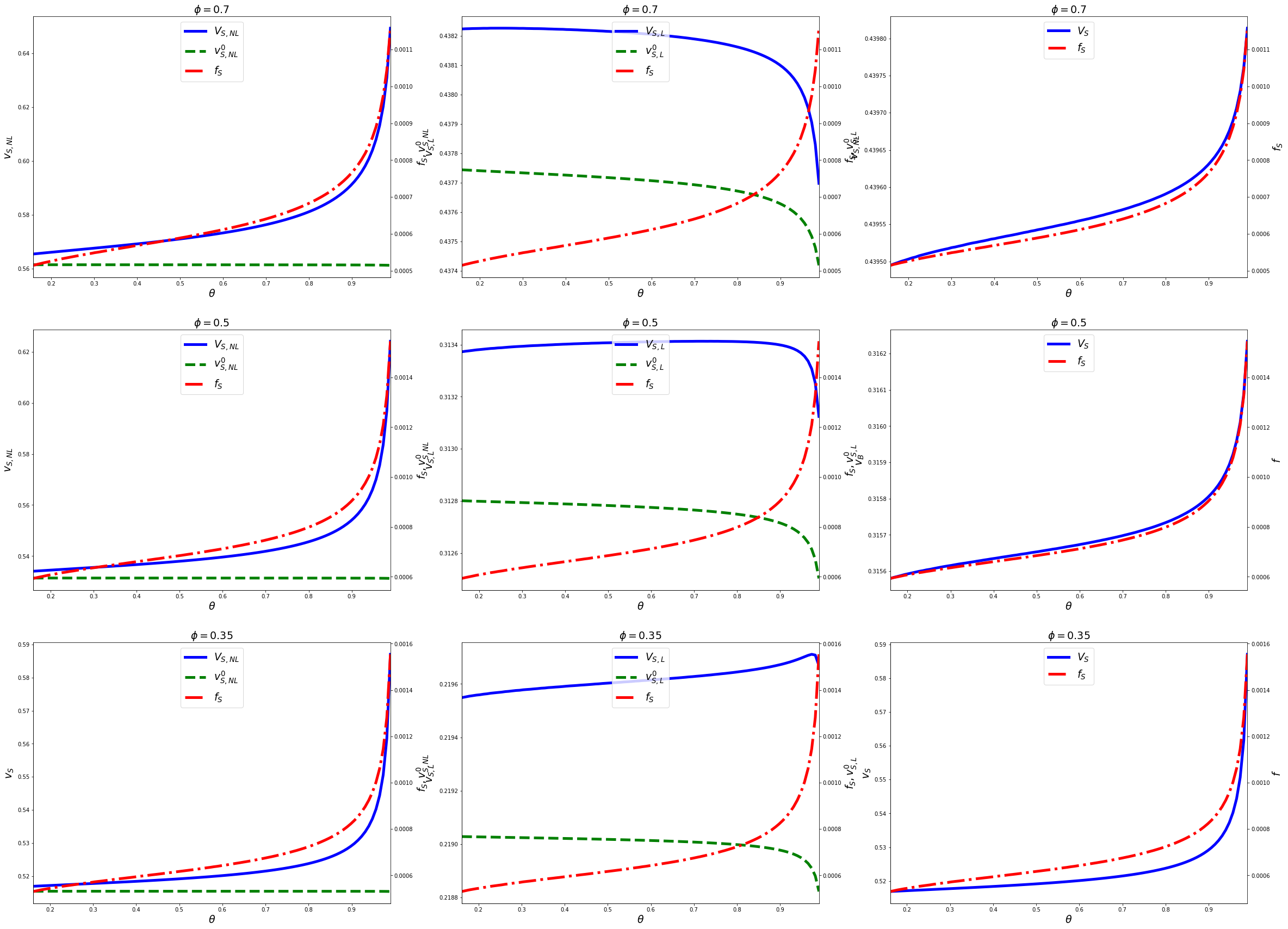}

\end{center}
\end{figure}

\begin{figure}[H]
\begin{center}\caption{Welfare of Sellers and Fee: $\pi = 0.1$}\label{Fig_sell}

\includegraphics[scale=0.15]{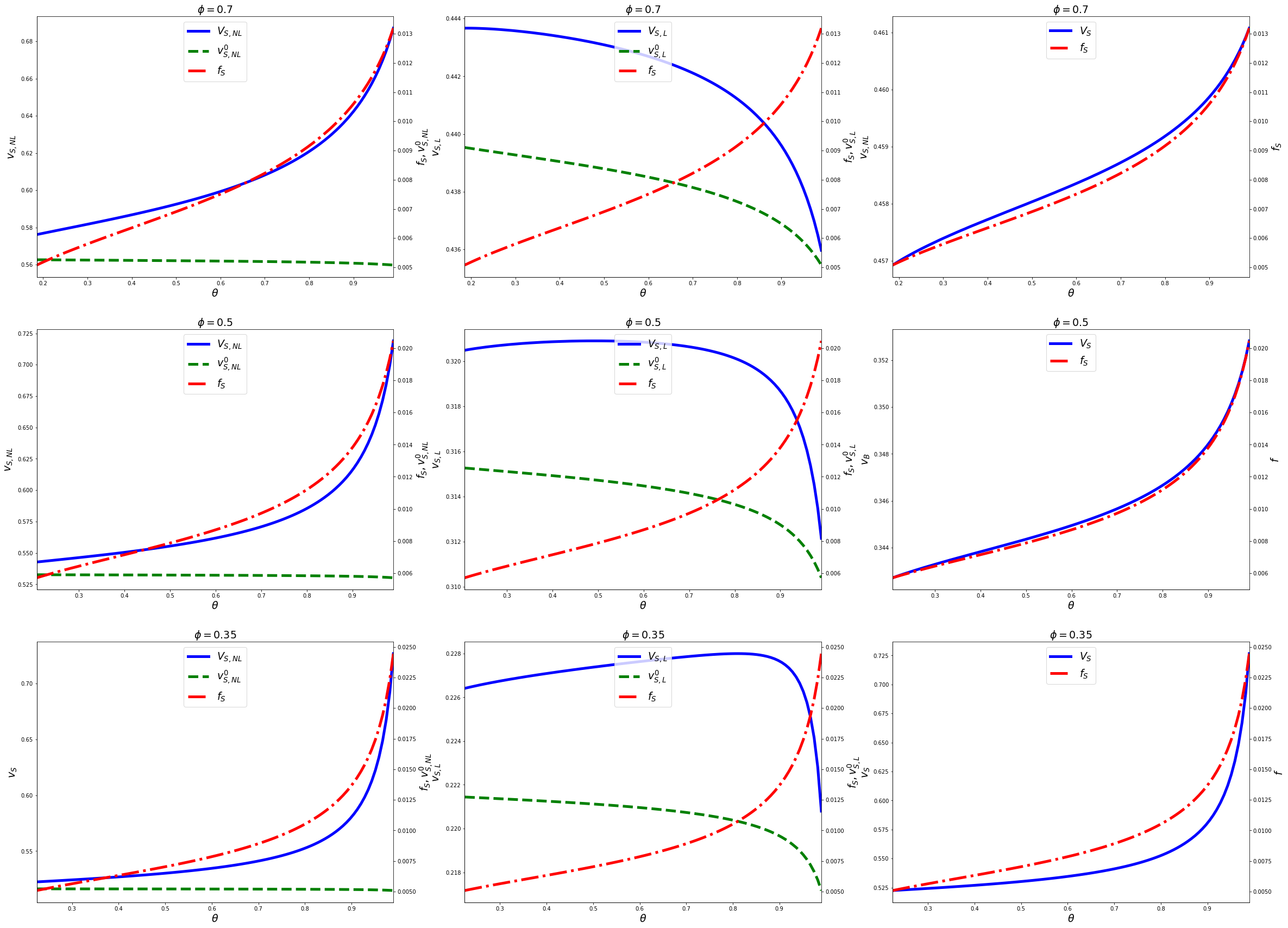}

\end{center}
\end{figure}

\begin{figure}[H]
\begin{center}\caption{Welfare of Sellers and Fee: $\pi = 0.4$}\label{Fig_sell}

\includegraphics[scale=0.15]{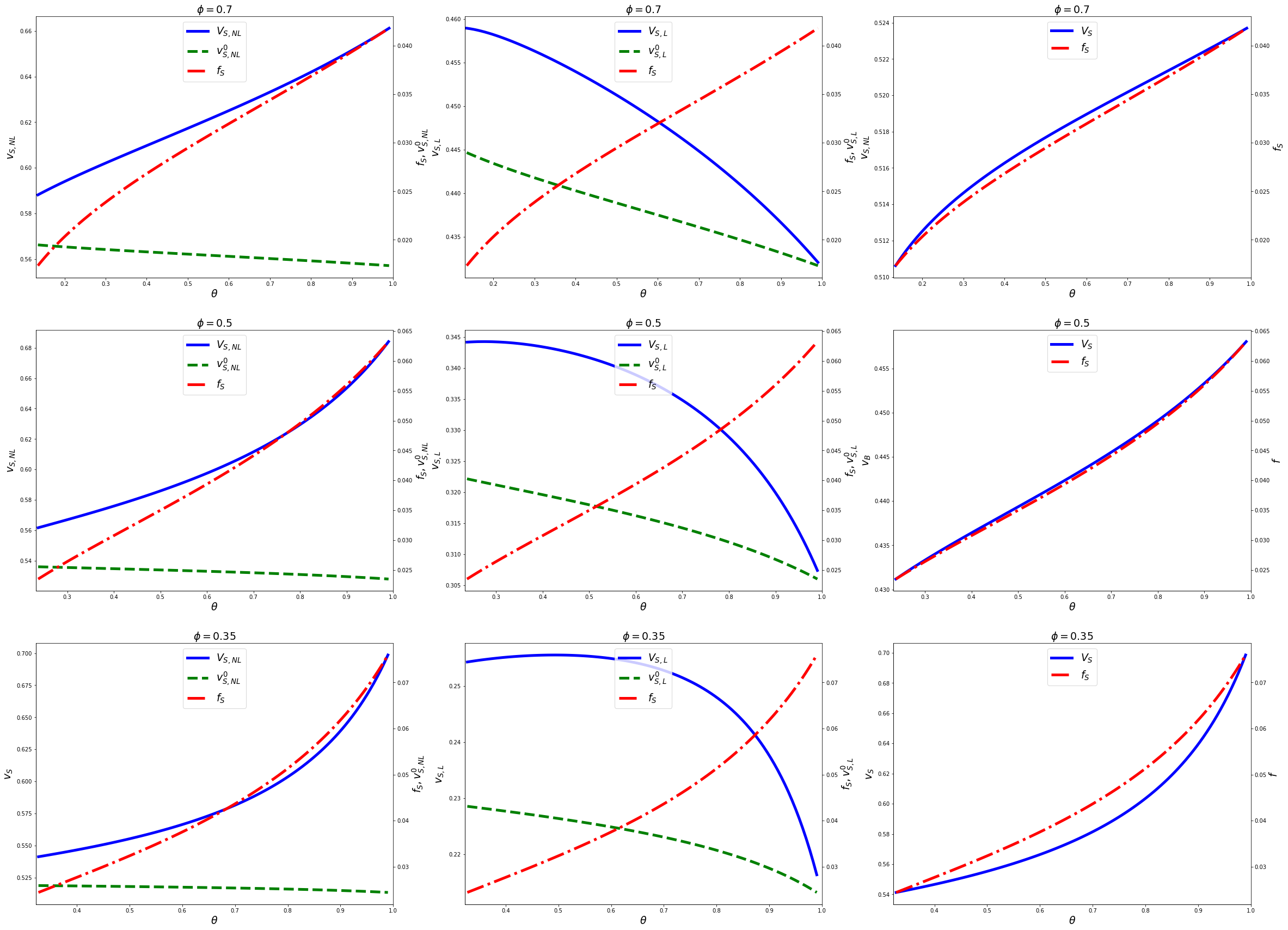}

\end{center}
\end{figure}

\begin{figure}[H]
\begin{center}\caption{Welfare of Sellers and Fee: $\pi = 0.7$}\label{Fig_sell}

\includegraphics[scale=0.15]{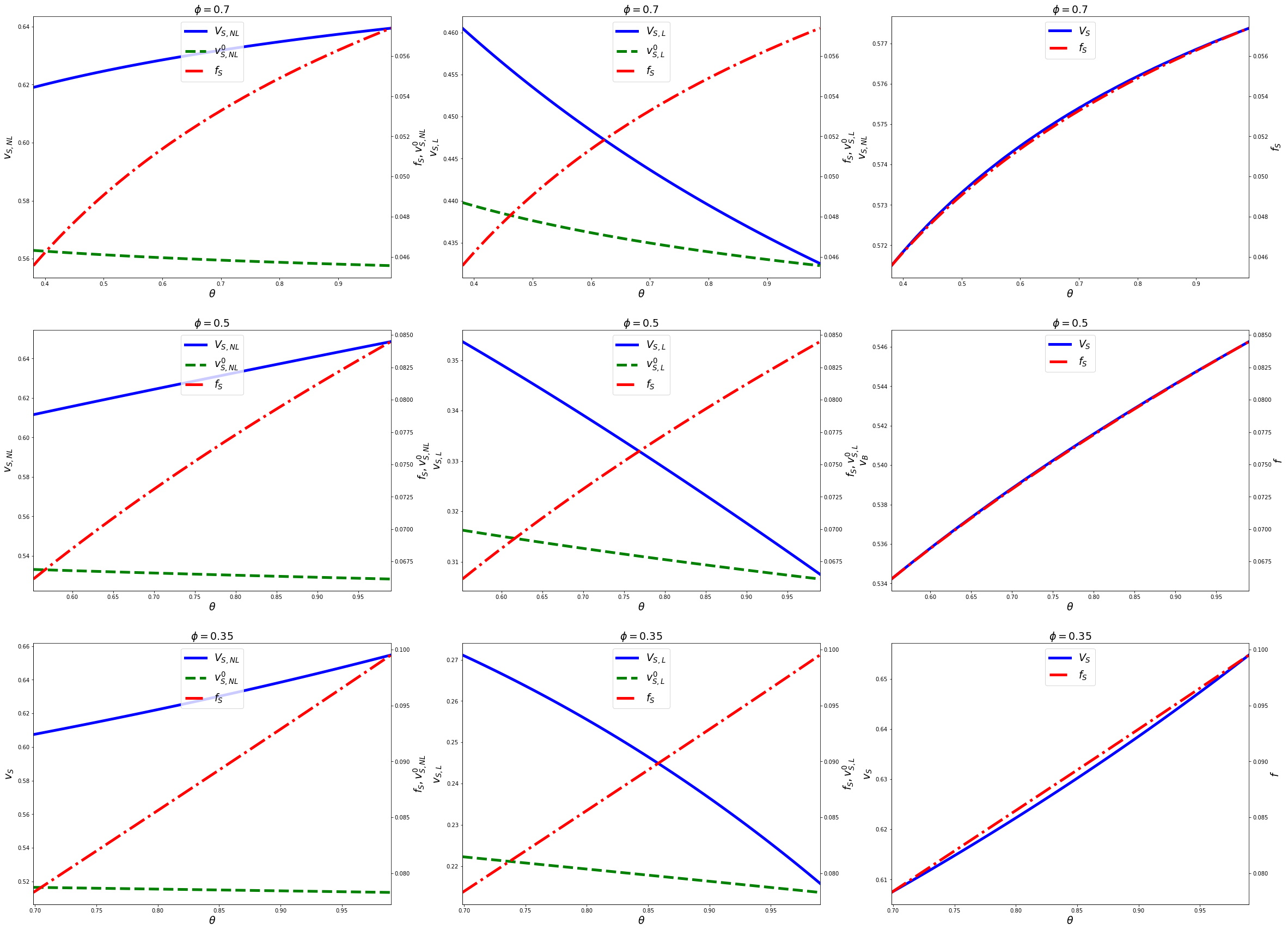}

\end{center}
\end{figure}

\begin{figure}[H]
\begin{center}\caption{Welfare of Sellers and Fee: $\pi = 0.9$}\label{Fig_sell}

\includegraphics[scale=0.15]{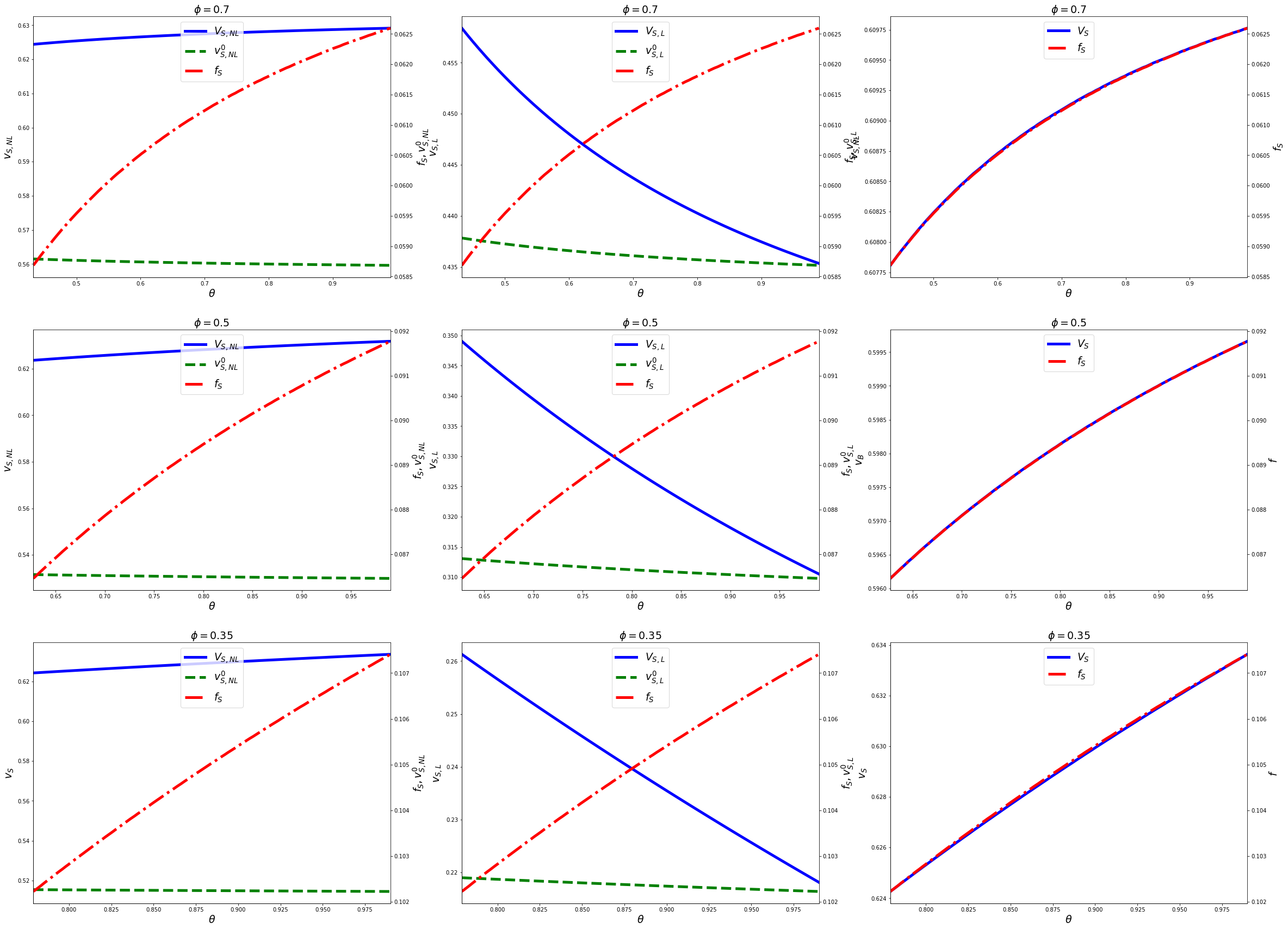}

\end{center}
\end{figure}

\end{document}